\newif\ifarxiv\arxivtrue
\def\year{2021}\relax
\documentclass[letterpaper]{article} 
\usepackage{aaai21}  
\usepackage{times}  
\usepackage{helvet} 
\usepackage{courier}  
\usepackage[hyphens]{url}  
\usepackage{graphicx} 
\usepackage{natbib}  
\usepackage{caption} 
\title{Fair Influence Maximization: A Welfare Optimization Approach} 
\author{
    Aida Rahmattalabi$^{\rm 1\star}$, Shahin Jabbari$^{\rm 2\star}$, Himabindu Lakkaraju$^{\rm 2}$, Phebe Vayanos$^{\rm 1}$\\ 
    Max Izenberg$^{\rm 3}$, Ryan Brown$^{\rm 3}$, Eric Rice$^{\rm 1}$, Milind Tambe$^{\rm 2}$
}
\affiliations{

    \textsuperscript{\rm 1}University of Southern California\\
    \textsuperscript{\rm 2}Harvard University\\ 
    \textsuperscript{\rm 3} {RAND} Corporation\\
    $^\star$ Equal Contribution. 
    \ifarxiv Correspondence to \href{mailto:rahmatta@usc.edu}{rahmatta@usc.edu} and \href{jabbari@seas.harvard.edu} {jabbari@seas.harvard.edu} \else Full version at \url{https://arxiv.org/abs/2006.07906}.
    \fi

}

\usepackage{hyperref}       
\usepackage{url}            
\usepackage{booktabs}       
\usepackage{amsfonts}       
\usepackage{nicefrac}       
\usepackage{pifont}
\newcommand{\xmark}{\ding{55}}%
\usepackage{microtype}      
\usepackage{color}
\usepackage{amsmath, amssymb, amsthm, amssymb, amsthm, amsmath, mathtools}
\usepackage{bm}
\usepackage{tikz}
\usepackage{caption}
\usepackage{subfigure}
\usepackage{multirow}
\usepackage{multicol}
\usepackage{wasysym}
\usepackage{graphicx}
\usepackage{ifdraft}
\usetikzlibrary{shapes,shapes.geometric,arrows,fit,calc,positioning,automata}

\newtheorem{lemma}{Lemma}
\newtheorem{proposition}{Proposition}
\newtheorem{corollary}{Corollary}
\newtheorem{remark}{Remark}

\def\reals{\mathbb{R}} 

\newcommand{\G}{\mathcal G}
\newcommand{\E}{\mathcal E}
\newcommand{\V}{\mathcal V}
\renewcommand{\k}{K}
\newcommand{\p}{p}
\renewcommand{\c}{c} 
\newcommand{\C}{\mathcal C} 
\newcommand{\NC}{C} 
\newcommand{\A}{\mathcal A} 
\newcommand{\bu}{\bm u} 
\renewcommand{\u}{u} 
\newcommand{\uc}{\u_\c} 
\newcommand{\wel}{W} 
\newcommand{\n}{N} 

\begin{document}	
\maketitle
\thispagestyle{plain}
\pagestyle{plain}
\begin{abstract}
Several behavioral, social, and public health interventions, such as suicide/HIV prevention or community preparedness against natural disasters, leverage social network information to maximize outreach. Algorithmic influence maximization techniques have been proposed to aid with the choice of ``peer leaders'' or ``influencers''
in such interventions. Yet, traditional algorithms for influence maximization have not been designed with 
these interventions in mind. As a result, they may disproportionately exclude minority communities from the benefits of the intervention. This has motivated research on \emph{fair influence maximization}. Existing techniques come with two major drawbacks. First, they require committing to a single fairness measure. Second, these measures are typically imposed as strict constraints leading to undesirable properties such as wastage of resources.
To address these shortcomings, we provide a principled characterization of the properties that a fair influence maximization algorithm should satisfy. In particular, we propose a framework based on social welfare theory, wherein the cardinal utilities derived by each community are aggregated using the isoelastic social welfare functions. Under this framework, the trade-off between fairness and efficiency can be controlled by a single \emph{inequality aversion} design parameter. We then show under what circumstances our proposed principles can be satisfied by a welfare function. The resulting optimization problem is monotone and submodular and can be solved efficiently with optimality guarantees. Our framework encompasses as special cases leximin and proportional fairness. Extensive experiments on synthetic and real world datasets including a case study on landslide risk management demonstrate the efficacy of the proposed framework.
\end{abstract}

\section{Introduction}
\label{sec:intro}
The success of many behavioral, social, and public health interventions relies heavily on effectively leveraging social 
networks~\cite{isaac2009gatekeeper,valente2007peer,TsangWRTZ19-groupfairness}. 
For instance, health interventions 
such as 
suicide/HIV prevention~\cite{gatekeeperyonemoto2019} and community preparedness against natural disasters involve finding a small set of 
well-connected individuals who can act as peer-leaders to detect warning signals 
(suicide prevention) or disseminate relevant information (HIV prevention or landslide risk management). 
The influence maximization framework has been employed to find such individuals~\cite{wilder2020clinical}. However,
such interventions 
may lead to discriminatory solutions
as individuals from racial minorities or LGBTQ communities may be disproportionately excluded from the benefits of the intervention~\cite{ RahmattalabiVFRWYT19,TsangWRTZ19-groupfairness}.

Recent work has incorporated fairness directly into influence maximization by proposing various notions of fairness such as maximin fairness~\cite{RahmattalabiVFRWYT19} and diversity constraints~\cite{TsangWRTZ19-groupfairness}.
Maximin fairness aims at improving the minimum amount of influence that any community receives. On the other hand, diversity 
constraints, inspired by the game theory literature, ensure that each community is at least as well-off had they received their 
share of resources proportional to their size and allocated them internally. Each of these notions offers a unique perspective on fairness. However, they also come with drawbacks. For example maximin fairness can result in significant degradation in total influence due to its stringent requirement to help the worst-off group as much as possible, where in reality it may be hard to spread the influence to some communities due to their sparse connections. On the other hand, while the diversity constraints aim at taking the community's ability in spreading influence into account, it does not explicitly account for reducing inequality (i.e., does not exhibit \emph{inequality aversion}). Consequently, there is no universal agreement on 
what fairness means and in fact, it is widely known that fairness is domain dependent~\cite{Narayanan18}. For example, excluding vulnerable communities from suicide prevention might have higher negative consequences compared to interventions promoting a healthier lifestyle.

Building on cardinal social welfare theory from the economics literature and principles of social welfare, we propose a principled characterization of the properties of social influence maximization solutions. In particular, we propose a framework for fair influence maximization based on social welfare theory, wherein the cardinal utilities derived by each community are aggregated using the isoelastic social welfare functions~\cite{bergson1938reformulation}. Isoelastic functions are in the general form of $ \u^{\alpha}/\alpha, \alpha < 1, \alpha \neq 0$ and $\log\u, \alpha = 0$ where $\alpha$ is a constant and controls the aversion to inequality and $\u$ is the utility value. They are used to measure the \emph{goodness} or \emph{desirability} 
of a utility distribution. However, due to the structural dependencies induced by the underlying social network, i.e., between-community and within-community edges, social welfare principles cannot be directly applied to our problem. Our contributions are as follows:
\emph{(i)} We extend the cardinal social welfare principles including the \emph{transfer principle} to the influence maximization framework, which is otherwise not applicable. We also propose a new principle which we call 
\emph{utility gap reduction}. 
This principle aims to avoid situations where high 
aversion to inequality leads to even more utility gap, caused by between-community influence spread; \emph{(ii)} We generalize the theory regarding these principles and show that for all problem instances, there does not exist a welfare function that satisfies all principles. Nevertheless, we show that if all communities are disconnected from one another (no between-community edges), isoelastic welfare functions satisfy all principles. This result highlights the importance of network structure, specifically between-community edges;~\emph{(iii)} 
Under this framework, the trade-off between fairness and efficiency can be controlled by a single \emph{inequality aversion} parameter $\alpha$. This allows a 
decision-maker to effectively trade-off quantities like utility 
gap and total influence by varying this parameter in the welfare function. 
We then incorporate these welfare functions as objective 
into an optimization problem to rule out undesirable solutions. We show that the resulting optimization problem is monotone and 
submodular and, hence, can be solved with a greedy algorithm with optimality guarantees; \emph{(iv)} Finally, we carry out detailed experimental analysis on synthetic and real social networks to study the trade-off between total influence spread and utility gap. In particular, we conduct a case study on the social network-based landslide risk management in Sitka, Alaska. We show that by choosing $\alpha$ appropriately we can flexibly control utility gap (4\%-26\%) and the resulting influence degradation (36\% - 5\%).

\noindent\textbf{Related Work. }
Recent work has incorporated fairness directly into the influence maximization framework by relying on either Rawlsian theory~\cite{Rawls09, RahmattalabiVFRWYT19}, game theoretic principles~\cite{TsangWRTZ19-groupfairness} or equality based notions~\cite{stoica2020seeding}. Equality based approaches strive for equal outcome across different communities. In general, strict equality is hard to achieve and may lead to wastage of resources. This is amplified in influence maximization as different communities have different capacities in being influenced (e.g., marginalized communities are hard to reach).
We discuss the other two approaches in more details in Sections~\ref{sec:existing-fairness}--\ref{sec:exp}. 
Social welfare functions have been used within the economic literature to study trade-offs between equality and efficiency~\cite{sen1997economic} and have been widely adopted in different decision making areas including health~\cite{abasolo2004exploring}.
Recently, \citet{Heidari2018fairness} proposed to study similar ideas for regression problems. The classical social welfare theory, however, does not readily extend to our setting due to dependencies induced by the between-community connections. Extending those principles is a contribution of our work. \ifarxiv See Appendix~\ref{sec:related-work}  \else See the full version \fi for a thorough review of related work. 

\section{Problem Formulation}
\label{sec:framework}

We use $\G = (\V, \E)$ to denote a graph (or network) in which $\V$ is the set of $\n$ vertices and $\E$ is the set of edges.
In the \emph{influence maximization problem}, a decision-maker chooses a set of at most $\k$ 
vertices to influence (or activate). The selected vertices then spread the influence in rounds according to 
the \emph{Independent Cascade Model}~\cite{KempeKT03}.\footnote{Our framework is also applicable to other forms of diffusion such as Linear Threshold Model~\cite{KempeKT03}} Under this model, each newly activated vertex spreads the influence to its neighbors independently and with a fixed probability $\p \in [0,1]$. The process 
continues until no new  vertices are influenced. We use $\A$ to denote the initial set of vertices, also referred to as influencer vertices. The goal of the decision-maker 
is to select a set $\A$ to maximize  the expected number of vertices that are influenced at the end of this process.
Each vertex of the graph belongs to one of the disjoint communities (empty intersection) $\c\in  \C := \{1, \ldots, \NC\}$ 
such that $\V_1 \cup \dots \cup \V_\NC = \V$ where $
\V_\c$ denotes the set of vertices that belong to community $\c$. This partitioning can be induced by, e.g., the intersection of a 
set of (protected) attributes such as race or gender for which fair treatment is important. We use $\n_\c$ to denote the size of community $\c$, i.e., $\n_\c = |\V_\c|$. Furthermore, communities may be  disconnected, in which case $\forall c,c'\in \C$ and $\forall v\in\V_\c, v'\in\V_{c'},$ there is no edge between $v$ and $v'$ (i.e., $(v,v') \notin \E$).
We define $\A^{\star} := \{\A\subseteq \V \mid |\A|\leq \k \}$ as the set of budget-feasible influencers. Finally, for any choice of influencers $\A \in \A^{\star}$, 
we let $\uc(\A)$ denote the utility, i.e., the expected fraction of the influenced vertices of community $\c$, where the expectation is taken over randomness in
the spread of influence. The standard influence maximization problem solves the optimization problem
\begin{equation}
        \displaystyle \mathop \text{maximize}_{\A \in \A^{\star} }  \displaystyle \sum_{\c \in\C} \n_\c\uc(\A). 
        \label{prob:standard-influence-maximization}
\end{equation}
When clear from the context, we will drop the dependency of $u_c(\A)$ on $\A$ to minimize notational overhead.

\section{Existing Notions of Fairness}
\label{sec:existing-fairness}
Problem~\eqref{prob:standard-influence-maximization} solely attempts to maximize the total influence which is also known as the \emph{utilitarian} approach.
Existing fair influence maximization problems are variants of Problem~\eqref{prob:standard-influence-maximization} involving additional constraints. We detail these below. See \ifarxiv Appendix~\ref{sec:related-work} \else the full version \fi for more discussion.

\noindent\textbf{Maximin Fairness (MMF). } Based on the
Rawlsian theory~\cite{Rawls09}, MMF~\cite{TsangWRTZ19-groupfairness}  aims to maximize the utility of the worst-off community. Precisely, MMF only allows $\A\in \A^{\star}$ that satisfy the following constraint
\begin{equation*}
\label{eq:mmf}
   \mathop \text{min}_{\c\in\C} \;\; \uc(\A) \geq \gamma,\;\; \text{ where }\A\in \A^\star,
\end{equation*}
where the left term is the utility of the worst-off community and $\gamma$ is the highest value for which the constraint is feasible.

\noindent\textbf{Diversity Constraints (DC). }
Inspired by the game theoretic notion of core, DC requires that every community 
obtains a utility higher than when it receives resources proportional to its size and allocates them internally~\cite{TsangWRTZ19-groupfairness}.
This is illustrated by the following constraint where $U_c$ denotes the maximum utility that community $\c$ can achieve with a budget equal to $\lfloor \k\n_\c/\n\rfloor$.
\begin{equation}
\label{eq:dc}
\uc(\A) \geq U_c,  \quad \forall \c\in\C \text{ where } \A \in \A^{\star}.
\end{equation}
DC sets utility lower bounds for communities based on their relative sizes and how well they can spread influence internally. As a result, it does not explicitly account for reducing inequalities and may lead to high influence gap. We show this both theoretically and empirically in Sections~\ref{sec:fair-connection} and~\ref{sec:exp}.

\noindent\textbf{Demographic Parity (DP). } Formalizing the legal doctrine of disparate impact~\cite{zafar2017fairness}, DP requires the utility of all communities to be roughly the same.
For any $\delta\in[0,1)$, DP implies the  constraints~\cite{AliBCMGS19, stoica2020seeding,aghaei2019learning}
\begin{equation*}
    \big|\u_\c(\A) -\u_{\c'}(\A)\big| \leq \delta,\;\; \forall \c,\c' \in \C  \text{ where } \A\in \A^{\star}.
\end{equation*}
The degree of tolerated inequality is captured by $\delta$ and higher $\delta$ values are associated with higher tolerance.  
We use exact and approximate DP to distinguish between $\delta = 0$ and $\delta > 0$. 

\section{Fair Influence Maximization }
\label{sec:welfare}

\subsection{Cardinal Welfare Theory Background}
Following the cardinal welfare theory~\cite{bergson1938reformulation}, our aim is to design welfare functions to measure the goodness of the choice of influencers. Cardinal welfare theory proposes a set of principles and welfare functions that are 
expected to satisfy these principles. Given two utility vectors, these principles  determine if they are indifferent or one of them is preferred. 
For ease of exposition, let $\wel$ denote this welfare function defined over the utilities of all individuals in the population (we will formalize $\wel$ shortly). Then the existing principles of social welfare theory can be summarized as follows.
Throughout this section, without loss of generality, we assume all utility vectors belong to $\mathbb [0,1]^N$.

\noindent\textbf{(1) Monotonicity. } 
If $\bu \prec \bu'$, then $\wel(\bu) < \wel(\bu')$.\footnote{$\prec$ means $\bu_c \leq \bu'_c$ for all $\c\in\C$ and $\bu_c < \bu'_c$ for some $\c\in \C$.}
In other words,  if $\bu'$ Pareto dominates $\bu$, then $\wel$ 
should strictly prefer $\bu'$ to $\bu$. This principle also appears  as  \emph{levelling  down}  objection  in 
political philosophy~\cite{Parfit97}. 

\noindent\textbf{(2) Symmetry. }
$\wel(\bu) = \wel\left(P(\bu)\right)$, where $P(\bu)$ is any element-wise permutation of $\bu$. According to this  principle, $\wel$ does 
not depend on the naming or labels of the individuals, but only on their utility levels.

\noindent\textbf{(3) Independence of Unconcerned Individuals. } Let $(\bu|^{\c}b)$ be a utility vector that is identical to $\bu$, 
except for the utility of individual $\c$ which is replaced by a new value $b$. The property requires that for all $\c, b,b',\bu$ and $\bu'$, 
$\wel(\bu|^{\c}b) < \wel(\bu'|^{\c}b) \Leftrightarrow \wel(\bu|^{\c}b') < \wel(\bu'|^{\c}b').$
Informally, this  principle states that $\wel$ should be independent of individuals whose utilities remain the same. 

\noindent\textbf{(4) Affine Invariance. } 
For any $a > 0$ and $b$, $\wel(\bu) < \wel(\bu') \Leftrightarrow \wel(a\bu + b) < \wel(a\bu' + b)$ 
i.e., the relative ordering is invariant to a choice of numeraire.

\noindent\textbf{($\star$5) Transfer Principle~\cite{Dalton20, Pigou12}. } Consider individuals  
$i$ and $j$ in utility vector $\bu$ such that $\u_{i} < \u_j$. Let $\bu'$ be another utility vector that is identical to $\bu$ in all 
elements except $i$ and $j$ where $\u'_i = \u_i + \delta$ and $\u'_{j} = \u_{j} - \delta$ for some 
$\delta \in (0, (\u_{j}-\u_{i})/2)$. Then, $\wel(\bu) < \wel(\bu')$. Informally, transferring 
utility from a high-utility to a low-utility individual should increase social welfare. 

It is well-known that any welfare function $\wel$ that satisfies the 
first four principles is additive and in the form of $\wel_\alpha(\bu) = \Sigma_{i=1}^{\n} \u_{i}^\alpha/\alpha$ for $\alpha \neq 0 $ and $\wel_\alpha(\bu) = \Sigma_{i=1}^{\n} \log(\u_{i})$ 
for $\alpha = 0$. Further, for $\alpha < 1$ the last principle is also satisfied. In this case $\alpha$ can be interpreted as an inequality aversion parameter, where smaller $\alpha$ values exhibit more aversion towards inequalities. We empirically investigate the effect of $\alpha$ in Section~\ref{sec:exp}.

\subsection{Group Fairness and New Principles}\label{sec:group-fair-principles}
Applying the cardinal social welfare framework to influence maximization problems comes with new challenges. We next highlight these challenges and demonstrate our approach.

First, the original framework of cardinal welfare theory defines the welfare function over individuals.  This is equivalent to seeking equality in the probability that each individual will be influenced, similar to the work of~\citet{FishBBFSV19}. It is notoriously hard to achieve individual fairness in practice, e.g., ~\citet{DworkHPRZ12} explore this in the machine learning context. The problem is further exacerbated in influence maximization because it is not always possible to spread the influence to isolated or poorly connected individuals effectively. Therefore, we focus on group fairness whereby the utility of each individual is defined as the average utility of the members of that community. Let $\uc$ denote the average utility of community $\c$. With this group-wise view, a welfare function can be written in terms of the average utilities over communities e.g., $\wel_\alpha(\bu) = \Sigma_{i=1}^{\n} \u_{i}^\alpha/\alpha = \Sigma_{\c\in\C} \n_\c\uc^\alpha/\alpha$.

\begin{figure}[ht!]
\centering
\includegraphics[width = 0.45 \textwidth]{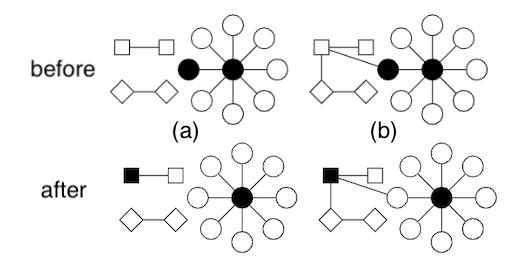}
\caption{The effect of network structure and in particular between-community edges on coupling of the utilities of communities. The figure shows two sample networks consisting of three communities, differentiated by shape: (a) is the same as (b) except that between-community edges are removed.  Black fillings show the choice of influencers. We further assume $p$ is small enough such that influence spread dissipates after one step. Transferring an influencer from circles to squares (top to bottom panel) affects the utility of diamonds in (b) but not in (a).}
\label{fig:network-structure}
\end{figure}

Moreover, while principles 1-4 can be easily extended to our influence maximization problem, this is not the case for the transfer principle. More precisely, in the influence maximization problem it might not be feasible to directly transfer utilities from one community to another without affecting the utilities of other communities. We highlight this effect with an example, see Figure~\ref{fig:network-structure}. In this figure, each 
community is represented by a distinct shape. The two networks (a) and (b) are identical except that between community edges are removed in network (a) (i.e., disconnected communities). The solid black vertices determine the choice of influencers. 
In network (b), if we transfer an influencer vertex from \emph{circles} to \emph{squares} according to Figure~\ref{fig:network-structure} (top to bottom panel), it will indirectly affect \emph{diamonds} as well. This effect is absent in network (a) as there are no between-community edges to allow the spread of influence across communities. In network (a), the transfer principle prefers the resulting utility vector after the transfer. However, this principle cannot be applied to network (b) as the utilities of more than one community is modified after the transfer. Additionally, even when direct transfer is possible, it can be the 
case that there is no symmetry in the amount of utility gained by low-utility community and the amount of utility lost by high-utility community after the transfer. To address both of these shortcomings we introduce the \emph{influence transfer principle} as a generalization of the transfer principle for influence maximization problems.

Similar to the original transfer principle, we consider solutions in which influencer vertices are transferred from one community to another community. Without loss of generality, we focus on the case where only one influencer vertex is transferred between the two communities. We refer to such solutions as neighboring solutions. Clearly, transfer of more than one influencer vertex can be seen as a sequence of transfers between neighboring solutions.

\noindent\textbf{(5) 
Influence Transfer Principle. } Let $\A$ and $\A'\in \A^{\star}$ be two neighboring solutions
with corresponding utility vectors
$\bu = \bu(\A)$ and $\bu' = \bu(\A')$. Suppose the elements of $\bu$ and $\bu'$ are sorted in 
ascending order. We also assume after the transfer, the ordering of the utilities stays the same across $\bu$ and $\bu'$.

If $\Sigma_{ \kappa \in \C : \kappa \leq \c} \n_{\kappa} (\u'_{\kappa} - \u_{\kappa}) \geq 0$ $\forall \c \in \C$  and 
$\u'_{\c} > \u_{\c}$ for some  $\c \in \C,$ then $\wel(\bu) < \wel(\bu')$.

Informally, influence transfer principle states that in a desirable transfer of utilities, the magnitude of the improvement in 
lower-utility communities should be at least as  high as the magnitude of decay in higher-utility communities while enforcing 
that at least one low-utility community receives a higher utility after the transfer. The original transfer is a special case of influence transfer principle when communities are disconnected and utilities transferred remain the same.

Next, we study whether any of the welfare functions that satisfy the first 4 principles satisfy the influence transfer
principle. In Proposition~\ref{pro:trans}, we show any additive and strictly concave function satisfies 
influence transfer principle. Since functions that satisfy the first 4 principles are strictly concave for $\alpha < 1$, the influence transfer principle is automatically satisfied in this regime. We defer all proofs to \ifarxiv Appendix~\ref{sec:omitted}. \else the full version. \fi
\begin{proposition}
\label{pro:trans}
Any strictly concave and additive function satisfies the influence transfer principle.
\end{proposition}

To measure inequality, notion of utility gap (or analogous notions such as ratio of utilities) is commonly used~\cite{FishBBFSV19,stoica2020seeding}. Utility gap measures the difference between the utilities of a pair of communities. In this work, we focus on the maximum utility gap, i.e., the gap between communities with the highest and lowest utilities (utility gap henceforth). For a utility vector $\bu$, we define \ $\Delta(\bu) = \max_{\c\in\C} \u_{\c}-\min_{\c\in\C} \u_{\c}$ to denote the utility gap. 
Fair interventions are usually motivated by the large utility gap before the intervention~\cite{gap}. 
\citet{FishBBFSV19} has shown that in social networks the utility gap can further increase 
after an algorithmic influence maximizing intervention. We extend this result to the entire class of welfare functions that we study in this work and we notice that the utility gap can increase even if we optimize for these welfare functions. This is a surprising result since, unlike the influence maximization objective, these welfare functions are designed to incorporate fairness, yet we may observe an increase in the utility gap. We now introduce another principle which aims to address this issue. Again we focus on 
neighboring solutions.

\noindent\textbf{(6) Utility Gap Reduction. }
Let $\A$ and $\A' \in \A^{\star}$ be two neighboring solutions with corresponding utility vectors 
$\bu = \bu(\A)$ and $\bu' = \bu(\A')$.
If $\Sigma_{\c\in\C} \n_\c\u_{\c} \leq \Sigma_{\c\in\C} \n_\c\u'_{\c}.$
and $\Delta(\bu) > \Delta(\bu')$ then  $\wel(\bu) < \wel(\bu')$.

The utility gap reduction simply states that the welfare function should prefer the utility vector whose total utility is at least as high as the other vector and also has smaller utility gap.
We now show that, in general, it is not possible to design a welfare function that obeys the utility gap reduction principle along with the other 
principles. 
\begin{proposition}
\label{pro:gap}
Let $\wel$ be a welfare function that obeys  principles 1-5. Then there exists an instance of influence maximization where $\wel$ does not satisfy the utility gap reduction.
\end{proposition}
Next, we show on a special class of networks, i.e., networks with disconnected communities, the utility gap reduction principle is satisfied in all influence maximization problems.
\begin{proposition}
\label{pro:disjoint}
Let $\wel$ be a welfare function that obeys principles 1-5. If the communities are disconnected, then $\wel$ also satisfies the utility gap reduction principle.
\end{proposition}
Propositions~\ref{pro:gap}~and~\ref{pro:disjoint} and their proofs establish new challenges in fair influence maximization. These challenges arise due to the coupling of the utilities as a result of the network structure and more precisely the between-community edges. The results in  Propositions~\ref{pro:gap}~and~\ref{pro:disjoint} leave open the following question: ``In what classes of networks, there exists a welfare function that satisfies all the 6 principles over all instances of influence maximization problems?"
As an attempt to answer this question, we empirically show that over various real and synthetic networks including stochastic block models, there exist welfare functions that obey all of our principles. We conclude this section by the following three remarks.
\begin{remark}[Application to Other Settings]
Our welfare-based framework can be \emph{theoretically} applied to different graph-based problems (e.g., facility location) but algorithmic solution is domain-dependent. The choice of influence maximization is motivated by evidence about discrimination studied in previous work~\cite{RahmattalabiVFRWYT19,stoica2020seeding}.
\end{remark}
\begin{table*}[ht!]
    \centering
    \begin{tabular}{cllllll}
    \toprule
    Notion/Principle & Monotonicity & Symmetry & Ind. of Unconcerned & Affine  &  Influence Transfer & Gap Reduction\\
    \toprule
    Exact DP & \xmark~(Prop.~\ref{pro:mono-dp-exact}) & \checkmark & \xmark~(Prop.~\ref{pro:dp-iuc}) & \checkmark  & \xmark~(Prop.~\ref{pro:dp-trans-2}) & \checkmark~(Prop.~\ref{pro:dp-ugr})\\
    Approximate DP & \xmark~(Prop.~\ref{pro:mono-dp}) & \checkmark & \xmark~(Prop.~\ref{pro:dp-iuc}) & \xmark  & \xmark~(Prop.~\ref{pro:dp-trans-2}) & \xmark~(Prop.~\ref{pro:dp-ugr})\\
    DC & \checkmark~(Cor.~\ref{cor:mono-dc-maximin}) & \xmark & \xmark~(Prop.~\ref{pro:dc-iuc}) & \xmark  & \xmark~(Prop.~\ref{pro:dc-trans}) & \xmark~(Prop.~\ref{pro:dc-ugr})\\
    MMF & \checkmark~(Cor.~\ref{cor:mono-dc-maximin})& \checkmark & \xmark~(Prop.~\ref{pro:lexi-iuc}) & \checkmark  & \xmark~(Prop.~\ref{pro:lexi-trans})  & \xmark~(Prop.~\ref{pro:lexi-ugr})\\ 
    Utilitarian  & \checkmark~(Cor.~\ref{cor:mono-dc-maximin}) & \checkmark & \checkmark & \checkmark & \xmark~(Prop.~\ref{pro:util-trans}) & \xmark~(Prop.~\ref{pro:util-ugr}) \\ 
    Welfare Theory & \checkmark  & \checkmark & \checkmark & \checkmark & \checkmark & \xmark~(Prop.~\ref{pro:gap})\\ 
    \bottomrule 

    \end{tabular}
    \caption{Summary of the properties of different fairness notions through the lens of welfare principles for influence maximization.}
    \label{tab:welfare-exisitng-notions}
\end{table*}
\begin{remark}[Relationship between Principles and Fairness]
Monotonicity ensures there is no wastage of utilities. Symmetry enforces the decision-maker to \emph{not} discriminate based on communities' names. According to the Independence of Unconcerned Individuals, between two solutions (choices of influencers) only those individuals/communities whose utilities change should impact the decision-maker's preference. Affine Invariance is a natural requirement that the preferences over different solutions should not change based on the choice of numeraire. Finally, the Transfer Principle promotes solutions that are more equitable. 
\end{remark}

\begin{remark}[Selecting the Inequality Aversion Parameter in Practice]
In our approach, $\alpha$ is a user-selected parameter that the user can vary to tune the trade-off between efficiency and fairness. Leaving the \emph{single} parameter $\alpha$ in the hands of the user is a benefit of our approach since the user can inspect the solution as $\alpha$ is varied to select their preferred solution. Since a single parameter must be tuned, this can be done without the need for a tailored algorithm. In particular, we recommend that $\alpha$ be either selected by choosing among a moderate number of values and picking the one with the most desirable behavior for the user or by using the bisection method. Typically, choosing $\alpha$ will reduce to letting the user select how much utility gap they are willing to tolerate: they will select the largest possible value of $\alpha$ for which the utility gap is acceptable. 
\end{remark}

\subsection{Group Fairness and Welfare Maximization}\label{sec:grou-fair-welfare}
The welfare principles reflect the preferences of a fair decision-maker between a pair of solutions. Thus a welfare function that satisfies all the principles would always rank the preferred (in terms of fairness and efficiency) solution higher. As a result, we can maximize the welfare function to get the most preferred solution. 

We show that the two classes of welfare functions $\wel_\alpha(\bu) = \Sigma_{i=1}^{\n} \u_{i}^\alpha/\alpha$ for $\alpha < 1, \alpha \neq 0 $ and $\wel_\alpha(\bu) = \Sigma_{i=1}^{\n} \log(\u_{i})$ 
for $\alpha = 0$ satisfy 5 of our principles.  Hence as a natural notion of fairness we can define a fair solution to be a choice of influencers with the highest welfare as defined in the following optimization problem.
\begin{equation}
\begin{array}{cl}
\displaystyle \mathop \text{maximize}_{\A\in \A^\star}  & \displaystyle \wel_{\alpha}(\bu(\A)). \\ 
\end{array}
\label{eq:IM-fair}
\end{equation}

\begin{lemma}
\label{lem:wel-sub-mon}
In the influence maximization problem, any welfare function that satisfies principles 1-5 is monotone and submodular.
\end{lemma}
It is well-known that to maximize any monotone submodular function, there exists a greedy algorithm 
with a $(1-1/e)$ approximation factor~\cite{NemhauserW81} which we can also use to solve the welfare maximization problem.

Each choice of the inequality aversion parameter $\alpha$ results in a different welfare function and hence a fairness notion. 
A decision-maker can directly use these welfare functions as objective of an optimization problem and study the trade-off 
between fairness and total utility by varying $\alpha$, see Section~\ref{sec:exp}.

\subsection{Connection to Existing Notions of Fairness}
\label{sec:fair-connection}
Our framework allows for a spectrum of fairness notions as a function of $\alpha$. It encompasses as a special case \emph{leximin fairness}\footnote{Leximin is subclass of MMF. According to its definition, among two utility vectors, 
leximin prefers the one where the worst utility is higher. If the worst utilities are equal, leximin repeats this process by comparing 
the second worst utilities and so on.}, a sub-class of MMF, for $\alpha\rightarrow -\infty$.
Proportional fairness~\cite{BonaldL2001-FairnessInternet}, a notion for resource allocation problems, is also closely connected to 
the welfare function  for $\alpha = 0$. See 
\ifarxiv Appendices~\ref{sec:related-work} and~\ref{sec:leximin}. \else the full version for more details. \fi

It is natural to ask which of the fairness principles are satisfied by the existing notions of fairness for influence maximization. As we discussed in Section~\ref{sec:existing-fairness}, the existing notions of fairness are imposed by adding constraints to the influence maximization problem. However, our welfare framework directly incorporates fairness into the objective. 
In order to facilitate the comparison, instead of the constrained influence maximization problems we consider an equivalent reformulation in which we bring the constraints into the objective via the characteristic function of the feasible set. We then have a single objective function which we can treat as the welfare function corresponding to the fairness constrained problem. 
More formally, given an influence maximization problem and fairness constraints written as a feasible set $\mathcal F$ 
\begin{equation*}
    \max_{\A\in\A^{\star}} \sum_{\c \in \C}\n_\c \u_\c(\A) \enspace\text{ s.t. }\enspace \bu(\A) \in \mathcal F.  
\end{equation*} 
We consider the following equivalent optimization problem
\begin{equation*}
\label{eq:wel-existing}
    \max_{\A\in\A^{\star}} \sum_{\c \in \C}\n_\c \u_\c(\A) + \mathcal I_{\mathcal F} (\bu(\A)) := \max_{\A\in\A^{\star}} \wel_{\mathcal F}(\bu(\A)),  
\end{equation*} 
in which $\mathcal I_{\mathcal F}(\bu)$ is equal to 0 if $\bu \in \mathcal F$ and $-\infty$ otherwise. Using this new formulation, we can now examine each of the existing notions of fairness though the lens of the welfare principles. Given the new interpretation, to show that a fairness notion does not satisfy a specific principle, it suffices to show there exist solutions $\A, \A'\in \A^\star$ and corresponding utility vectors $\bu = \bu(\A)$ and $\bu'=\bu(\A')$ such that the principle prefers $\bu$ over $\bu'$ but $\wel_{\mathcal F}(\bu) < \wel_{\mathcal F}(\bu')$.
The results are summarized in Table~\ref{tab:welfare-exisitng-notions} where in addition to comparing with the previous notions introduced in Section~\ref{sec:existing-fairness}, we compare with the utilitarian notion i.e., Problem~\eqref{prob:standard-influence-maximization}. We provide formal proofs for each entry of Table~\ref{tab:welfare-exisitng-notions} in  \ifarxiv Appendix~\ref{sec:table}. \else the full version.\fi

We observe that none of the previously defined notions of fairness for influence maximization satisfies all of our principles and each existing notion violates 
at least 3 out of 6 principles. 
We point out that exact DP is the only notion that satisfies the utility gap reduction. However, this comes at a cost as enforcing exact DP may result in significant reduction in total utility in the fair solution compared to the optimal unconstrained solution~\cite{corbett2017algorithmic}.
\begin{figure*}[ht!]
\centering
\begin{subfigure}
\centering
\includegraphics[width = 0.35\linewidth]{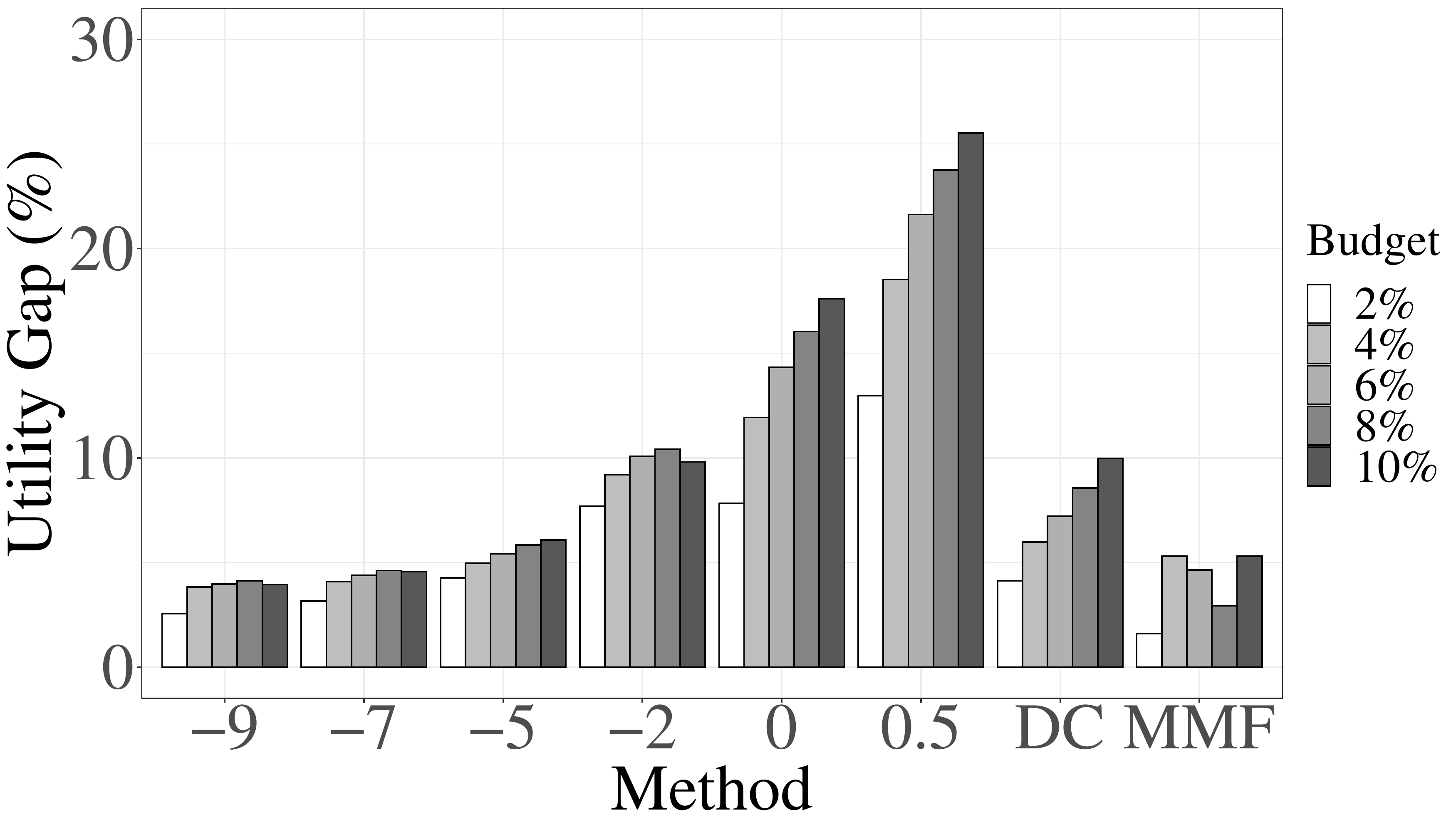}
\end{subfigure}%
\begin{subfigure}
\centering
\includegraphics[width = 0.35\textwidth]{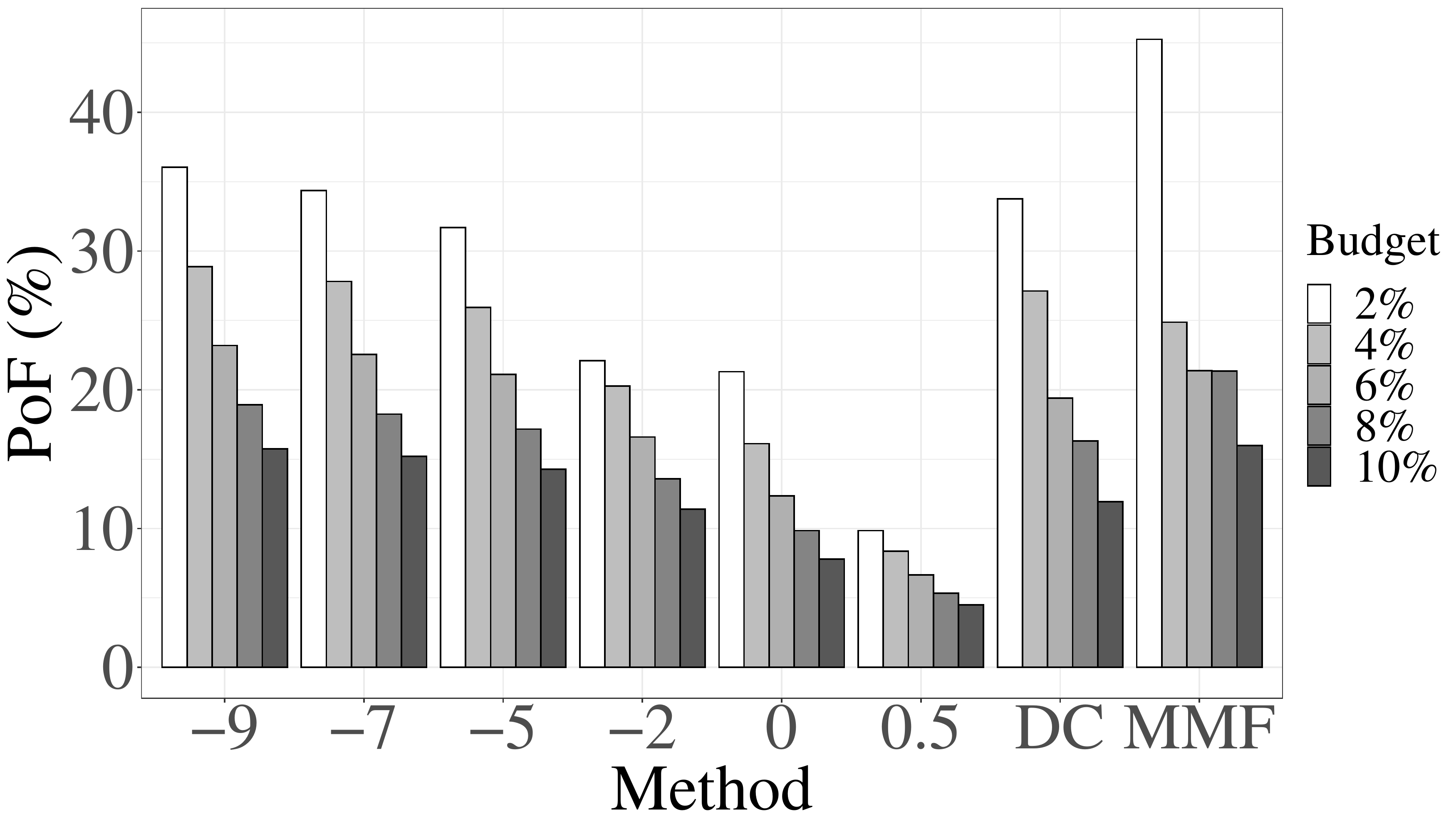}
\end{subfigure}%
\caption{Left and right panels: utility gap and PoF for different $\k$ and $\alpha$ values for our framework and baselines.}
\label{fig:Sitka-all}
\end{figure*}
\begin{figure}
\centering
\includegraphics[width = 0.34\textwidth]{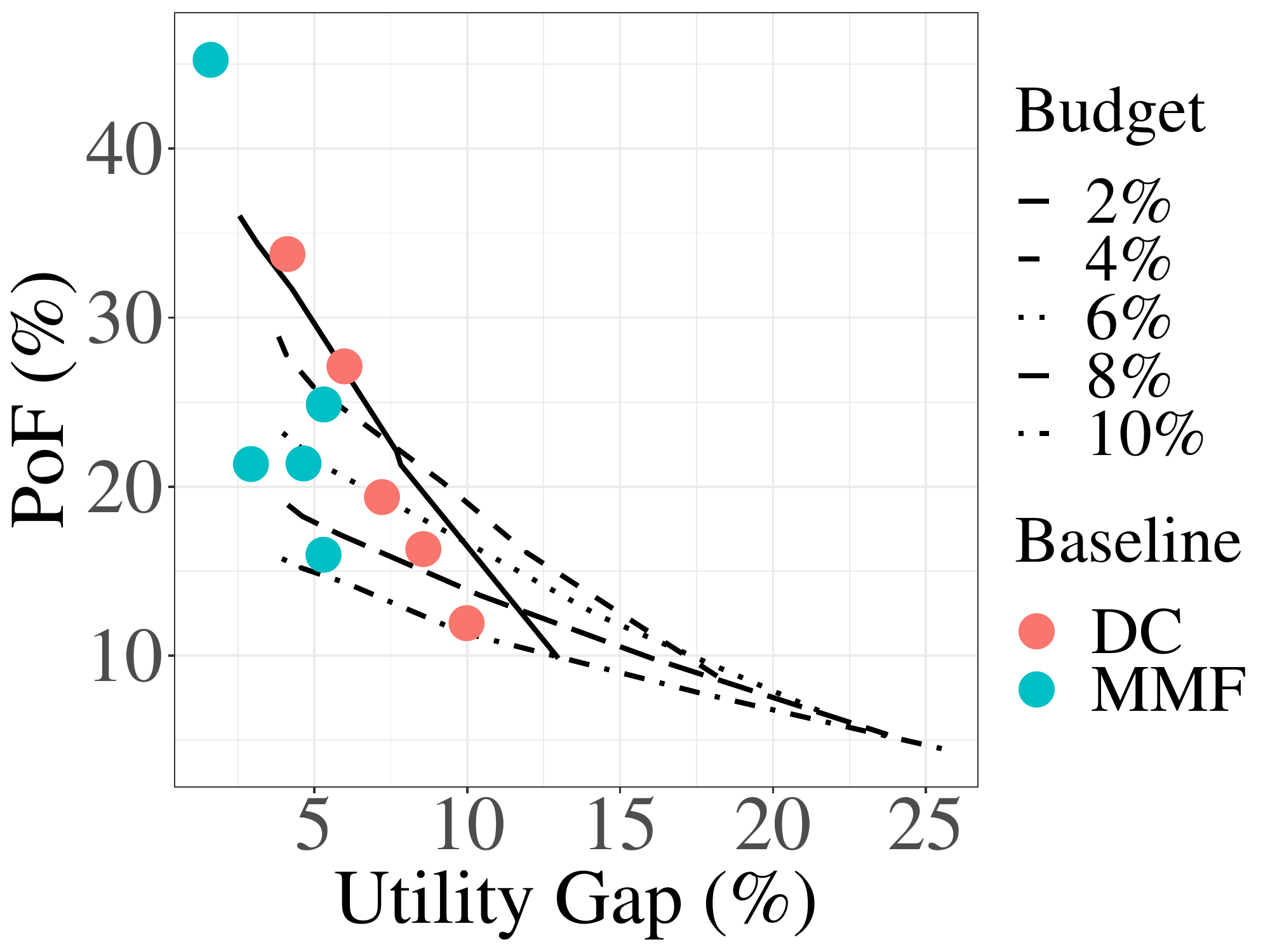}
\caption{PoF vs. utility gap trade-off curves. Each line corresponds to a different budget $\k$ across different $\alpha$ values.}
\label{fig:trade}
\end{figure}
\section{Computational Results}
\label{sec:exp}
We evaluate our approach in terms of both the total utility or spread of influence (to account for efficiency) and utility gap (to account for fairness). We show by changing the inequality aversion parameter, we can effectively trade-off efficiency with fairness.  As baselines, we compare with DC and MMF. To the best our knowledge, there is no prior work that handles DP constraints over the utilities. We follow the approach of~\citet{TsangWRTZ19-groupfairness} for both problems and view these problems as a multi-objective submodular optimization with utility of each each community being a separate objective. They propose an algorithm and implementation with asymptotic $(1-1/e)$ 
approximation guarantee which we also utilize here.  
We use \emph{Price of Fairness} (PoF), defined as the percentage loss in the total influence spread as a measure of efficiency. Precisely, $\text{\rm{PoF}} := 1 - \text{\rm{OPT}}^{\text{\rm{fair}}}/\text{\rm{OPT}}^{\text{\rm{IM}}}$ in
which $\text{\rm{OPT}}^{\text{\rm{fair}}}$ and $\text{\rm{OPT}}^{\text{\rm{IM}}}$ are the the total influence spread, with and without fairness. Hence PoF$\in[0,1]$ and smaller values are more desirable. The normalization in  PoF allows for a meaningful comparison between networks with different sizes and budgets as well as between different notions of fairness. In the PoF calculations, we utilize the generic greedy algorithm~\cite{KempeKT03} to compute $\text{\rm{OPT}}^{\text{\rm{IM}}}$.  To account for fairness, we compare the solutions in terms of the utility gap. Analogous measures are widely used in fairness literature~\cite{HardtPS16-EOPP} and more recently in graph-based problems~\cite{FishBBFSV19,stoica2020seeding}. We also note that our framework ranks solutions based on their welfare and does not directly optimize utility gap, as such our evaluation metric of fairness does not favor any particular approach.

We perform experiments on both synthetic and real networks. We study two applications: community preparedness against landslide incidents and suicide prevention among homeless youth. We discuss the latter in \ifarxiv Appendix~\ref{sec:omitted-exp}. \else the full version. \fi In the synthetic experiments, we use the \emph{stochastic block model} (SBM) networks, a widely used model for networks with community structure~\cite{FienbergS81}. In SBM networks, vertices 
are partitioned into disjoint communities. Within each community $\c$, an edge between two vertices is present independently with probability $q_\c$. Between any two vertices in communities $\c$ and  $\c'$, an edge exists independently with probability $q_{\c\c'}$ and typically 
$q_\c > q_{\c\c'}$ to capture homophily~\cite{mcpherson2001birds}. SBM captures the community structure of social networks~\cite{girvan2002community}. We report the average results over 20 random instances and set $\p=0.25$ in all experiments.

\noindent\textbf{Landslide Risk Management in Sitka, Alaska.  }
Sitka, Alaska is subject to frequent landslide incidents. In order to improve communities' preparedness, an effective approach is to instruct people on how to protect themselves before and during landslide incidents. Sitka has a population of more than 8000 and instructing everyone is not feasible. Our goal is to select a limited set of individuals as peer-leaders to spread information to the rest of the city. The Sitka population is diverse including different age groups, political views, seasonal and stable residents where each person can belong to multiple groups. These groups differ in their degree of connectedness. This makes it harder for some groups to receive the intended information and also impacts the cost of imposing fairness. We investigate these trade-offs.

Since collecting the social network data for the entire city is cumbersome, we assume a SBM network and use in-person semi-structured interview data from 2018-2020 with members of Sitka to estimate the SBM  parameters. Using the interview responses in conjunction with the voter lists, we identified 5940 individuals belonging to 16 distinct communities based on the intersection of age groups, political views, arrival time to Sitka (to distinguish between stable and transient individuals). 
The size of the communities range from 112 (stable, democrat and 65+ years of age) to 693 (republican, transient fishing community, age 30-65). See \ifarxiv Appendix~\ref{sec:omitted-exp} \else the full version \fi for details on the estimation of network parameters.

Figure~\ref{fig:Sitka-all} summarizes results across different budget values $\k$ ranging from $2\%$ to $10\%$ of the network size $\n$ for our framework (different $\alpha$ values) as well as the baselines. In the left panel, we observe that as $\alpha$ decreases, our welfare-based framework further reduces the utility gap, achieving lower gap than DC and competitive gap as MMF. As we noted in Section~\ref{sec:fair-connection}, our framework recovers leximin (which has stronger guarantees than MMF) as $\alpha \rightarrow -\infty$, though we show experimentally that this is achieved with moderate values of $\alpha$. Overall, utility gap shows an increasing trend with budget, however the sensitivity to budget decreases when more strict fairness requirements are in place, e.g. in MMF and $\alpha = -9.0$.
From the right panel, PoF varies significantly across different approaches and budget values surpassing $40\%$ for MMF. This is due to the stringent requirement of MMF to raise the utility of the worst-off as much as possible. Same holds true for lower values of $\alpha$ as they exhibit higher aversion to inequality. The results also indicate that PoF decreases as $\k$ grows which captures the intuition that fairness becomes less costly when resources are in greater supply. Resource scarcity is true in many practical applications, including the landslide risk management domain which makes it crucial for decision-makers to be able to study different fairness-efficiency trade-offs to come up with the most effective plan. Figure~\ref{fig:trade} depicts such trade-off curves where each line corresponds to a different budget value across the range of $\alpha$. Previous work only allows a decision-maker to choose among a very limited set of fairness notions regardless of the application requirements. Here, we show that our framework allows one to choose
$\alpha$ to meaningfully study the PoF-utility gap trade-offs. For example, given a fixed budget and a tolerance on utility gap, one can choose an $\alpha$ with the lowest PoF. 
We now investigate the effect of relative connectedness. We study the effect of relative community size in \ifarxiv Appendix~\ref{sec:exp-relative}. \else the full version. \fi

\begin{figure}[ht!]
\centering
\begin{subfigure}
\centering
\includegraphics[width = 0.32\textwidth]{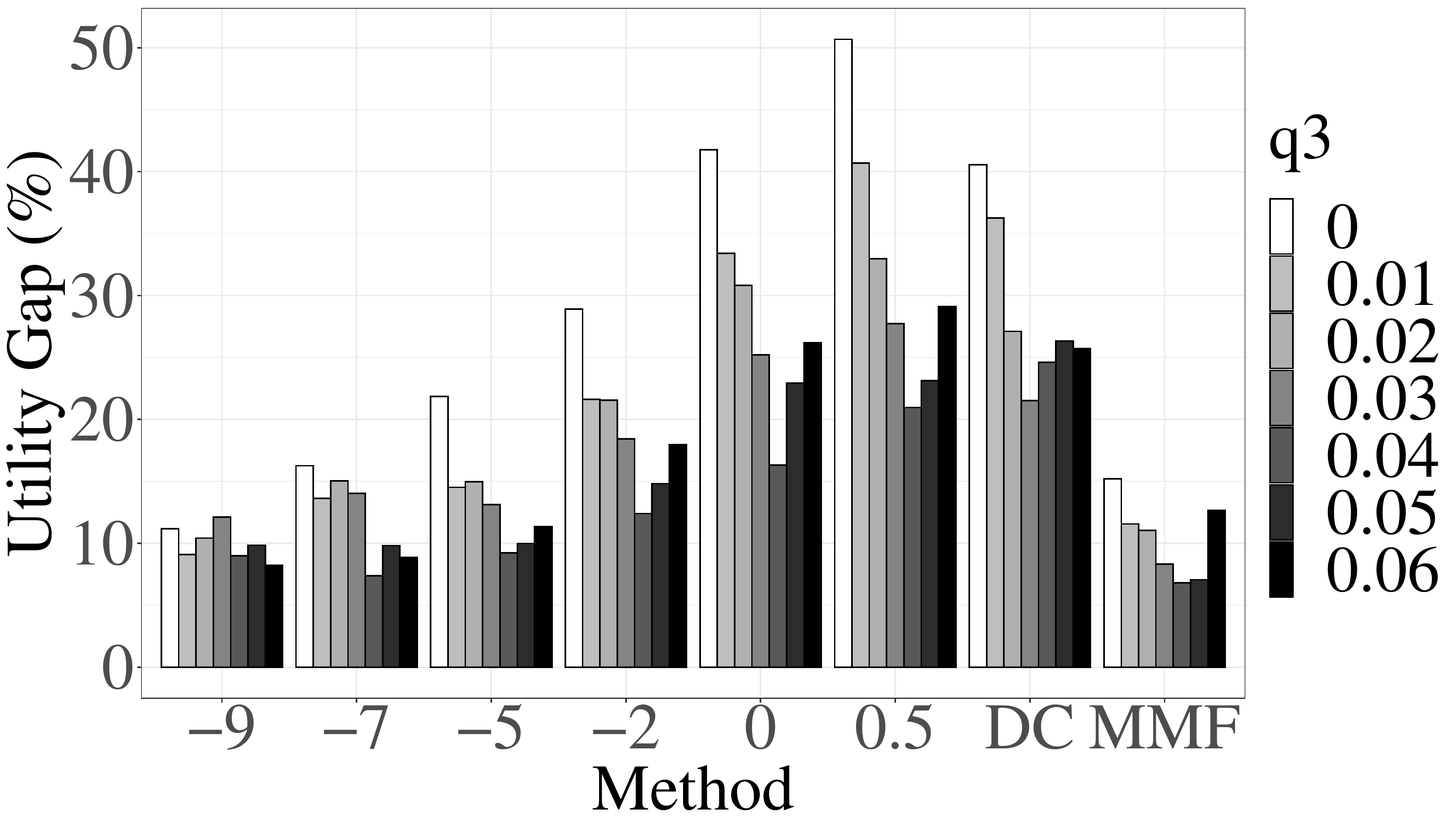}
\end{subfigure}%
\begin{subfigure}
\centering
\includegraphics[width = 0.32\textwidth]{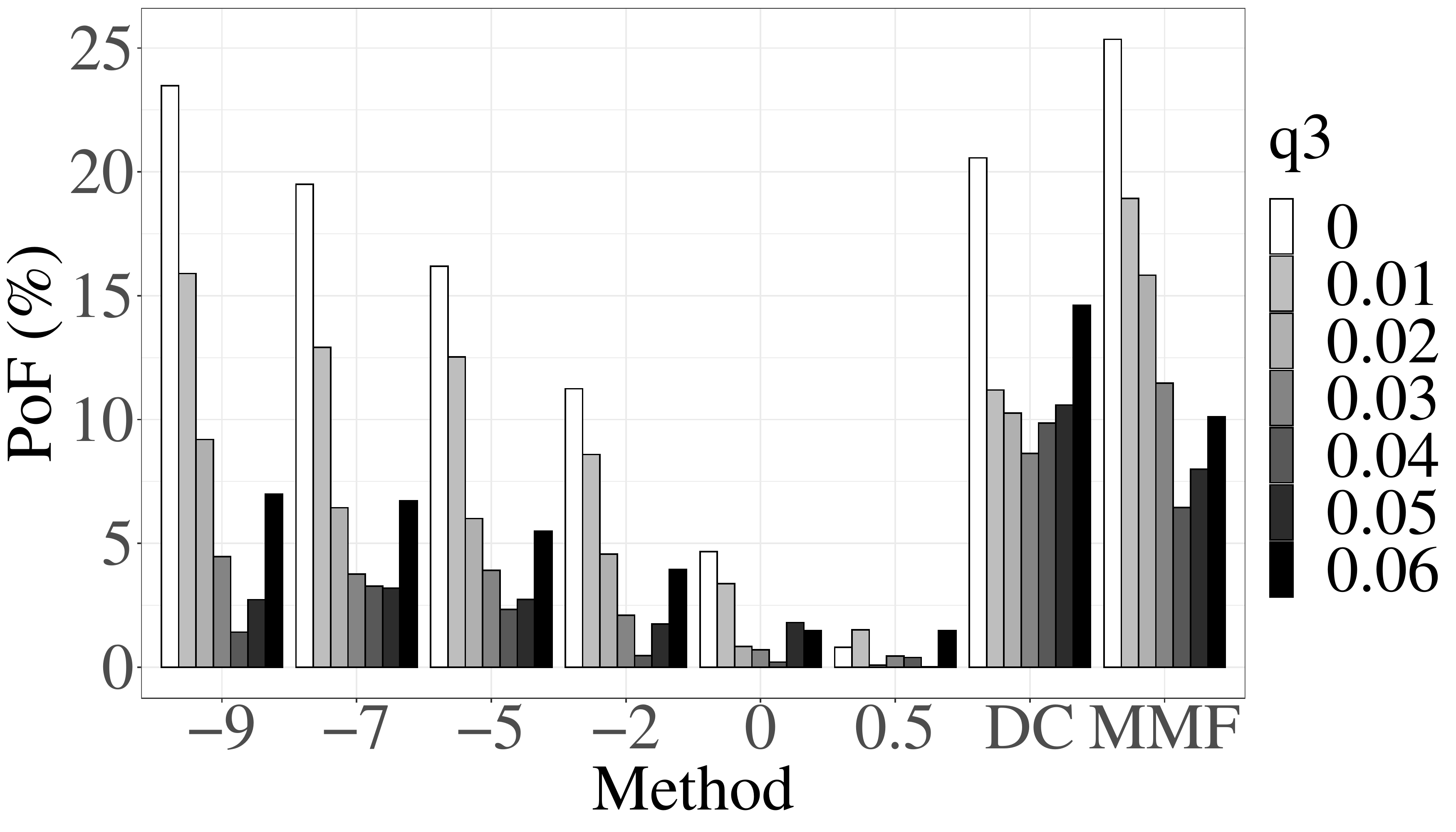}
\end{subfigure}
\caption{Utility gap and PoF for various levels $q_3$. All results are compared across different values of $\alpha$ and the baselines.}
\label{fig:synthetic-all-connect}
\end{figure}
\noindent\textbf{Relative Connectedness. } We sample SBM networks consisting of 3 communities each of size 100 where communities differ in their degree of 
connectedness. We set $q_1 = 0.06, q_2 = 0.03, q_3 = 0.0$ to obtain three communities with high, mid and low relative connectedness. We choose these values to reflect asymmetry in the structure of different communities which mirrors real world scenarios since not every community is 
equally connected. We set  between-community edge probabilities $q_{\c\c'}$ to 0.005 for all $\c$ and $\c'$ and $\k = 0.1 \n$. We gradually increase $q_3$ from 0.0 to 0.06. Results are summarized in Figure~\ref{fig:synthetic-all-connect}, where each group of bars correspond to a different approach. We observe as $q_3$ increases utility gap and PoF decrease until they reach a minimum around at around $q_3 = 0.03$. From this point, the trend reverses. This U-shaped behavior is due to structural changes in network. 
More precisely, for $q_3 < 0.03$ we are in the high-mid-low connectedness regime for the three groups, where the third community receives the minimum utility. 
As a result, as $q_3$ increases it becomes more favorable to choose more influencer vertices in this community which in turn reduces the utility gap. For $q_3 > 0.03$, the second community will be become the new worst-off community due its lowest connectedness. Hence, further increase in $q_3$ causes more separation in connectedness and we see previous behavior in reverse. 
Thus, by further increasing $q_3$, communities 1 and 3 receive more and more influencer vertices. This behavior translates to PoF as the relative connectedness of communities impacts how \emph{hard} it is to achieve a desired level of utility gap. Finally, we see that the U-shaped behavior is skewed, i.e., we observe higher gap and PoF in lower range of $q_3$ which is due to higher gap in connectedness of communities.

We can also compare the effect of relative connectedness and community size (see \ifarxiv Appendix~\ref{sec:exp-relative}. \else the full version.\fi) We observe that connectedness has a more significant impact on PoF (up to 25\%) compared to community size (less than 4\%). In other words, when communities are structurally different it is more costly to impose fairness. This is an insightful result given that in different applications we may encounter different populations with structurally different networks. Utility gap on the other hand is affected by both size and connectedness.

Finally while our theory indicates that in the network setting, no welfare function can satisfy all principles including utility gap reduction over all instances of the influence maximization, we observe that our class of welfare functions satisfies all of the desiderata on the class of networks that we empirically study. Our theoretical results showed this for a special case of networks with disconnected communities. In particular, we see higher PoF is accompanied by lower utility gap which complies with utility gap reduction principle.

\newpage\clearpage\clearpage
\section*{Ethics Statement and Broader Impact}
As the empirical evidence highlighting ethical side effects of algorithmic decision-making is growing~\cite{propublica, hiring}, the nascent field of algorithmic fairness has also witnessed a significant growth. 
It is well-established by this point that there is no universally agreed-upon notion of fairness, as fairness concerns vary from one domain to another~\cite{Narayanan18, BerkHJJKMNR17}. The need for different fairness notions can also be explained by theoretical studies that show that different fairness definitions are often in conflict with each other~\cite{KleinbergMR17, Chouldechova17, FriedlerSV16}. To this end, most of the literature on algorithmic fairness proposes different fairness notions motivated by different ethical concerns. A major drawback of this approach is the difficulty of comparing these methods against each other in a systematic manner to choose an appropriate notion for the domain of interest. Instead of following this trend, we propose a unifying framework controlled by a single parameter that can be used by a decision-maker to systematically compare different fairness measures which typically result in different (and possibly also problem-dependent) trade-offs. Our framework also accounts for the social network structure while designing fairness notions -- a consideration that is mainly overlooked in the past. Given these two contributions, it is perceivable that our approach can be used in many of the public health interventions such as suicide, HIV or Tuberculous prevention that rely on social networks. This way, the decisions-makers can compare a menu of fairness-utility trade-offs proposed by our approach and decide which one of these trade-offs are more desirable without a need to understand the underlying mathematical details that are used in deriving these trade-offs. 

There are crucial considerations when deploying our system in practice. First, cardinal welfare is one particular way of formalizing fairness considerations. This by no means implies that other approaches for fairness e.g. equality enforcing interventions should be completely ignored. Second, we have assumed that the decision-maker has the full knowledge of the network as well as possibly protected attributes of the individuals which can be used to define communities. Third, while our experimental evaluation is based on utilizing a greedy algorithm, it is conceivable that this greedy approximation can create complications by imposing undesirable biases that we have not accounted for. Intuitively (and as we have seen in our experiments) the extreme of inequality aversion ($\alpha\rightarrow -\infty$) can be used as a proxy for pure equality. However, the last two concerns require more care and we leave the study of such questions as future work.
\vspace{-5mm}
\paragraph{Acknowledgement} We would like to thank David Gray Grant, Matthew Joseph, Andrew Perrault and Bryan Wilder for helpful discussions about this work. This work is supported in part by the Smart \& Connected Communities program of the National Science Foundation under NSF award No. 1831770, the US Army Research Office under grant number W911NF1710445 and Center for Research on Computation and Society (CRCS) at the Harvard John A. Paulson School of Engineering and Applied Sciences.
\bibliography{paper}

\newpage\clearpage
\appendix
\section{Additional Related Work}
\label{sec:related-work}

Artificial Intelligence and machine learning algorithms hold great promise in addressing many pressing societal problems. These problems often pose complex ethical and fairness issues which need to be addressed before the algorithms can be deployed in the real world.
The nascent field of algorithmic fairness has emerged to address these fairness concerns. 
To this end, different notions of fairness are defined based on 
one or more \emph{sensitive attributes} such as age, race or gender. For example, in the classification and regression setting, these notions mainly aim at equalizing a statistical quantity across different 
communities or populations~\cite{HardtPS16-EOPP,zafar2017fairness}. 
While surveying the entirety of this field is out of our 
scope (see e.g.,~\cite{BerkHJKR18} for a recent survey), we point out that there is a wide range of fairness notions defined across different settings 
and it has been shown that the right notion is problem dependent~\cite{BerkHJJKMNR17,Narayanan18} and also different notions of fairness 
can be incompatible with each other~\cite{KleinbergMR17}. Thus, care must be taken when we employ these notions of fairness across different applications.

Motivated by the importance of fairness when conducting interventions in social initiatives~\cite{kube2019allocating}, fair influence maximization 
has received a lot of attention recently ~\cite{TsangWRTZ19-groupfairness,AliBCMGS19, RahmattalabiVFRWYT19, FishBBFSV19}.  These works 
have incorporated fairness directly into the influence maximization framework by (1) relying on either Rawlsian theory of 
justice~\cite{Rawls09, RahmattalabiVFRWYT19}, (2) game theoretic principles~\cite{TsangWRTZ19-groupfairness} or (3) equality 
based notions~\cite{AliBCMGS19, stoica2020seeding}. We will discuss the first two approaches in more details in 
Sections~\ref{sec:existing-fairness}~and~\ref{sec:welfare}, as well as in our experimental section. 
Equality based approaches strive for equal outcomes across different communities. In general, strict equality is hard to achieve and may lead to wastage of resources. This is amplified in influence maximization as different communities have different capacities in being influenced (e.g., marginalized communities are hard to reach). 

\citet{FishBBFSV19} investigate the notion of information access gap, where they propose maximizing the minimum probability 
that an individual is being influenced/informed to constrain this gap. As a result they study fairness at an individual level while we 
study fairness at the group level. Also, their notion of access gap is limited to the gap in a bipartition of the network which is in 
principle different from utility gap that we study which accommodates arbitrary number of protected groups. 

Similar to our work,~\citet{AliBCMGS19} also  study utility gap. They propose an optimization model that directly penalizes utility 
gap which they solve via a surrogate objective function.  Their surrogate functions are in the form of a sum of concave functions of the 
group utilities which are aggregated with arbitrary weights. Unlike their work, our approach takes an axiomatic approach with strong 
theoretical justifications and it does not allow for arbitrary concave functions and weights as they violate the welfare principles.

There has also been a long line of work considering fairness in resource allocation problems 
(see e.g.,~\cite{BertsimasFT11-pricefair,kleinbergRT1999-fairness,BonaldL2001-FairnessInternet,biswas2018fair}). More recently, group 
fairness has been studied in the context of resource allocation problems~\cite{conitzer2019group, elzayn2019fair,benabbou2018diversity} 
and specifically in graph covering problems~\cite{RahmattalabiVFRWYT19}.
In resource allocation setting, \emph{maximin fairness} and \emph{proportional fairness} are widely adopted fairness notions. Proportional 
fairness is a notion introduced for bandwidth allocation~\cite{BonaldL2001-FairnessInternet}. An allocation is proportionally fair if the sum 
of percentage-wise changes in the utilities of all groups cannot be improved with another allocation. In classical resource allocation problems, 
each individual or group has a utility function that is independent of the utilities of others individuals or groups. However, this is not the case in 
influence maximization due to the underlying social network structure i.e., the between-community edges which makes our problem distinct from the classical resource allocation problems. We note that, while in the bandwidth allocation setting there is also a network structure, the utility of each vertex is still independent of the other vertices and is only a function of the amount of resources that the 
vertex receives.

Finally, \citet{Heidari2018fairness} have recently proposed to study inequality aversion and welfare through cardinal welfare theory in the context of regression problems. Their main contribution is to use this theory to draw attention to other fairness considerations beyond equality. However, the classical social welfare theory, does not readily extend to our setting due to dependencies induced by the between-community connections. Indeed, extending those principles is a contribution of our work. 
\section{Omitted Proofs from Section~\ref{sec:group-fair-principles}}
\label{sec:omitted}
\begin{proof}[Proof of Proposition~\ref{pro:trans}]
Let $F:\reals^{N}\to\reals$ be an additive function in the form $F(\bu) = \sum^{N}_{i=1} f(\u_i)$ where $f:\reals \to \reals$ is a monotonically increasing and strictly concave function. We are focusing on group fairness where the utility of each individual is given by the average utility of their community. Hence, we can rewrite $F(\bu) = \sum_{\c \in \C} \n_\c f(\u_\c)$.
Let $\bu = \bu(\A)$ and $\bu' = \bu(\A')$ denote the utility vectors corresponding to neighboring solutions $\A$ and $\A'$, respectively. Suppose $\bu$ and $\bu'$ are sorted in ascending order and for all $\c\in\C$, index $\c$ in both vectors corresponds to the same community, i.e., after the transfer the ordering of the utilities has not changed.

Furthermore, assume 
$\Sigma_{ \kappa \in \C : \kappa \leq \c} \n_{\kappa} (\u_{\kappa} - \u'_{\kappa}) \geq 0$, $\forall \c \in \C$  and 
$\u_{\c} > \u'_{\c}$   for some  $\c \in \C$. Clearly $\bu$ and $\bu'$ satisfy the assumptions of the influence transfer principle. 
We need to show that $\Sigma_{\c\in\C} \n_\c f(\u_{\c}) > \Sigma_{\c\in\C} \n_\c f(\u'_{\c})$ or $\Sigma_{\c\in\C} \n_\c \left(f(\u_{\c})- f(\u'_{\c})\right) > 0$.

The proof is by induction. We iteratively sweep the vectors $\bu$ and $\bu'$ from the smallest index to the largest and show that for any  $\kappa\in\C$, $\Sigma_{\c\leq \kappa} \n_\c \left(f(\u_{\c})-f(\u'_{\c})\right) \geq 0$ with inequality becoming strict for at least one 
$\kappa$. To do so we repeatedly use a property of strictly concave functions known as decreasing marginal returns. According to this 
property $f(x+\delta_x)-f(x) >  f(y+\delta_y)-f(y)$ for $x < y$ and $\delta_x \geq \delta_y > 0.$

\begin{figure*}[ht!]
\centering
\begin{tikzpicture}
[scale=0.55, red node/.style={circle,fill=gray, draw=black}, blue node/.style = {rectangle, fill = gray, draw=black, scale = 1},
green node/.style = {diamond, fill = gray, draw=black, scale = 0.8}
, scale=1.0, every
edge/.style={->,-> = latex'},
]
\tikzset{scale=0.8, edge/.style = {-,> = latex'}}
\tikzset{scale=0.8, directededge/.style = {->,> = latex'}}

\node[red node] (1) at  (1.41, 0){};
\node[red node] (6) at  (1.41, 1.41){};
\node[red node] (9) at  (1.41, -1.41){};
\node[red node] (10) at  (5, 0){};
\node[red node] (11) at  (3, 0){};
\node[red node] (12) at  (3.59, 1.41){};
\node[red node] (13) at  (5, 2){};
\node[red node] (14) at  (6.41, 1.41){};
\node[red node] (15) at  (7, 0){};
\node[red node] (16) at  (6.41, -1.41){};
\node[red node] (17) at  (5,-2){};
\node[red node] (18) at  (3.59, -1.41){};

\draw[edge] (10) to (11);
\draw[edge] (10) to (12);
\draw[edge] (10) to (13);
\draw[edge] (14) to (10);
\draw[edge] (15) to (10);
\draw[edge] (16) to (10);
\draw[edge] (17) to (10);
\draw[edge] (18) to (10);

\node[blue node] (31) at  (9.41, 0){};
\node[blue node] (32) at  (9.41, 1.41){};
\node[blue node] (33) at  (9.41, -1.41){};
\node[blue node] (19) at  (11.41, 0){};
\node[blue node] (20) at  (11.41, 1.41){};
\node[blue node] (21) at  (11.41, -1.41){};
\node[blue node] (22) at  (15, 0){};
\node[blue node] (23) at  (13, 0){};
\node[blue node] (24) at  (13.59, 1.41){};
\node[blue node] (25) at  (15, 2){};
\node[blue node] (26) at  (16.41, 1.41){};
\node[blue node] (27) at  (17, 0){};
\node[blue node] (28) at  (16.41, -1.41){};
\draw[edge] (22) to (23);
\draw[edge] (22) to (24);
\draw[edge] (22) to (25);
\draw[edge] (22) to (26);
\draw[edge] (22) to (27);
\draw[edge] (22) to (28);

\node[blue node] (40) at  (21, 0){};
\node[blue node] (41) at  (19, 0){};
\node[blue node] (42) at  (19.59, 1.41){};
\node[blue node] (43) at  (21, 2){};
\node[blue node] (44) at  (22.41, 1.41){};
\draw[edge] (40) to (41);
\draw[edge] (40) to (42);
\draw[edge] (40) to (43);
\draw[edge] (40) to (44);

\node[green node] (47) at  (24.41, 0){};
\node[green node] (48) at  (24.41, 1.41){};
\node[green node] (49) at  (24.41, -1.41){};
\node[green node] (50) at  (26.41, 0){};
\node[green node] (51) at  (26.41, 1.41){};
\node[green node] (52) at  (26.41, -1.41){};
\node[green node] (53) at  (28.41, 0){};
\node[green node] (54) at  (28.41, 1.41){};
\node[green node] (55) at  (28.41, -1.41){};
\node[green node] (56) at  (32, 0){};
\node[green node] (57) at  (30, 0){};
\node[green node] (58) at  (30.59, 1.41){};
\node[green node] (59) at  (32, 2){};
\node[green node] (60) at  (33.41, 1.41){};
\node[green node] (67) at  (34, 0){};
\draw[edge] (56) to (57);
\draw[edge] (56) to (58);
\draw[edge] (56) to (59);
\draw[edge] (56) to (60);
\draw[edge] (56) to (67);

\node[red node] (61) at  (38, 0){};
\node[red node] (62) at  (36, 0){};
\node[red node] (63) at  (36.59, 1.41){};
\node[green node] (64) at  (38, 2){};
\node[red node] (65) at  (38, 0){};
\node[green node] (66) at  (39.41, 1.41){};

\draw[edge] (61) to (62);
\draw[edge] (61) to (63);
\draw[edge] (61) to (64);
\draw[edge] (61) to (65);
\draw[edge] (61) to (66);
\end{tikzpicture}
\caption{An illustration for the graph used in the proof of Proposition~\ref{pro:gap} without the correct scaling. There are 
three communities (circle, square and diamond) and they all have size 100. The circle community consists of an ``all-circle" 
star structure with 80 vertices, 14 isolated vertices and a mixed star structure (shared with the diamond community) with 6 
circle vertices. The square community consists of two ``all-square" star structures with sizes 60 and 10 plus a set of 30 isolated 
vertices. The diamond community consists of an ``all-diamond" star structure with 30 vertices, 66 isolated vertices and a mixed 
star structure (shared with the circle community) with 4 diamond vertices. }
\label{fig:pop-sep}
\end{figure*}
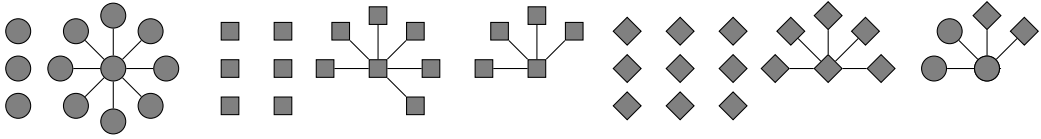

More specifically, in our inductive step, we keep track of a  ``decrement budget'' which we denote by $\Delta$. Intuitively if we can show 
that  $\Sigma_{\c\leq \kappa} \n_\c\left(f(\u_{\c})-f(\u'_{\c})\right) > 0$ with budget $\Delta$ for some $\kappa$, we can then use the 
decreasing marginal return property along with the assumption that $\bu'$ is sorted to show that as long as 
$\n_{\kappa+1}\left(\u'_{\kappa+1} - \u_{\kappa+1}\right) \leq \Delta$ it is the case that $\Sigma_{\c\leq \kappa+1} \n_\c\left( f(\u_{\c})-f(\u'_{\c})\right) > 0$. 
After each round we update the $\Delta$ and move on to the next element in the utility vectors.

Formally, let $\Delta = 0$ to start at the begining of this inductive process. After visiting the $\c$th community, we simply update 
$\Delta$ by $\Delta \leftarrow \Delta + \n_\c\left(\u_{\c}-\u'_{\c}\right)$. By the assumption of the transfer principle $\Delta$ is 
non-negative at all points of this iterative process and is strictly positive at some point during the process. Observe that 
$f(\u_1) \geq f(\u'_1)$ since $\u_1\geq \u'_1$ by the assumption of the transfer principle and monotonicity of $f$. We can use this as the base case. 
Since $\bu$ and $\bu'$ are sorted, given that $\Delta$ is non-negative,
the fact that $f$ is strictly concave (so that the decreasing marginal 
return property can be used) immediately implies that $\Sigma_{\c\leq \kappa} \n_\c\left(f(\u_{\c})-f(\u'_{\c})\right) \geq 0$ at any 
iteration $\kappa$ of the process. The inequality becomes strict for some $\kappa$ given the assumption of the transfer principle. 
This proves the claim.
\end{proof}

\begin{proof}[Proof of Proposition~\ref{pro:gap}]
Figure~\ref{fig:pop-sep} is an illustration of the graph that is used in the proof to witness the statement. We set $\p=1$ 
(deterministic spread) and number of initial seeds $\k=4$. Consider two choices of influencer vertices $\A$ and $\A'$. Let $\A$ denote the choice 
that consists of the center of all-star structures that consist of a single community. Let $\A'$ denote the solution that is identical 
to $\A$ with the sole difference that only the center of one of the all-square structures is chosen and the last seed is selected to be 
the center of the star structure that is the mix of circle and diamond communities. Clearly these two solutions are neighboring. 
The average utilities for these solutions are $(\text{diamond = } 0.3, \text{square = }0.7, \text{circle = }0.8)$ in $\bu$ and 
$ (\text{diamond = }0.34, \text{square = }0.6, \text{circle = }0.86)$ in $\bu'$, respectively. Both solutions correspond to a total 
utility of 180  but the utility gap is $\Delta(\bu)=0.5$ for $\bu$ as opposed to the utility gap of $\Delta(\bu')=0.52$ for $\bu'$. So a welfare 
function that obeys the utility gap reduction should prefer $\bu$ over $\bu'$.

\begin{figure}[ht!]
\centering
\begin{subfigure}
\centering
\includegraphics[width=2.6in]{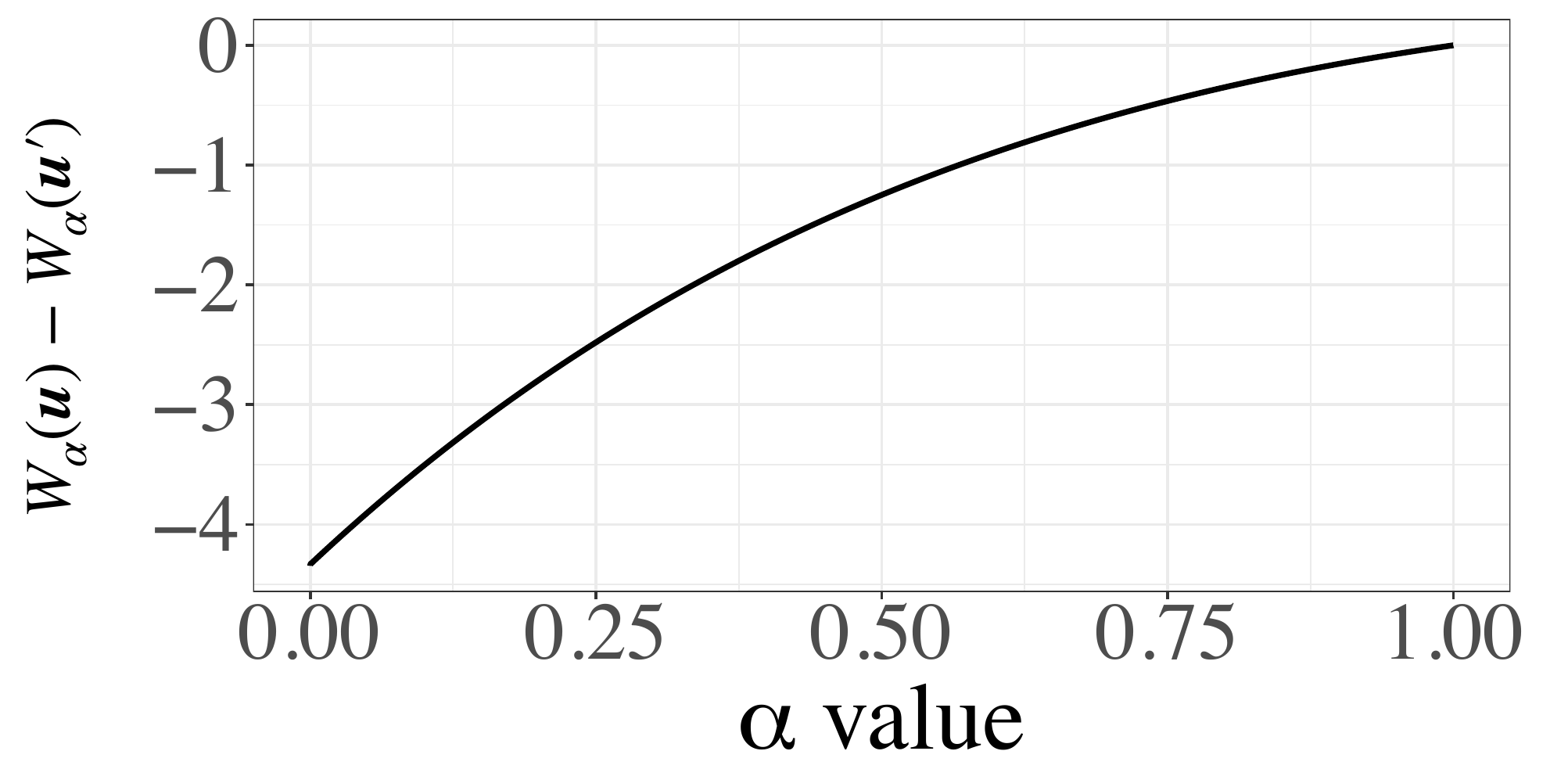}
\end{subfigure}
\quad
\begin{subfigure}
\centering
\includegraphics[width=2.6in]{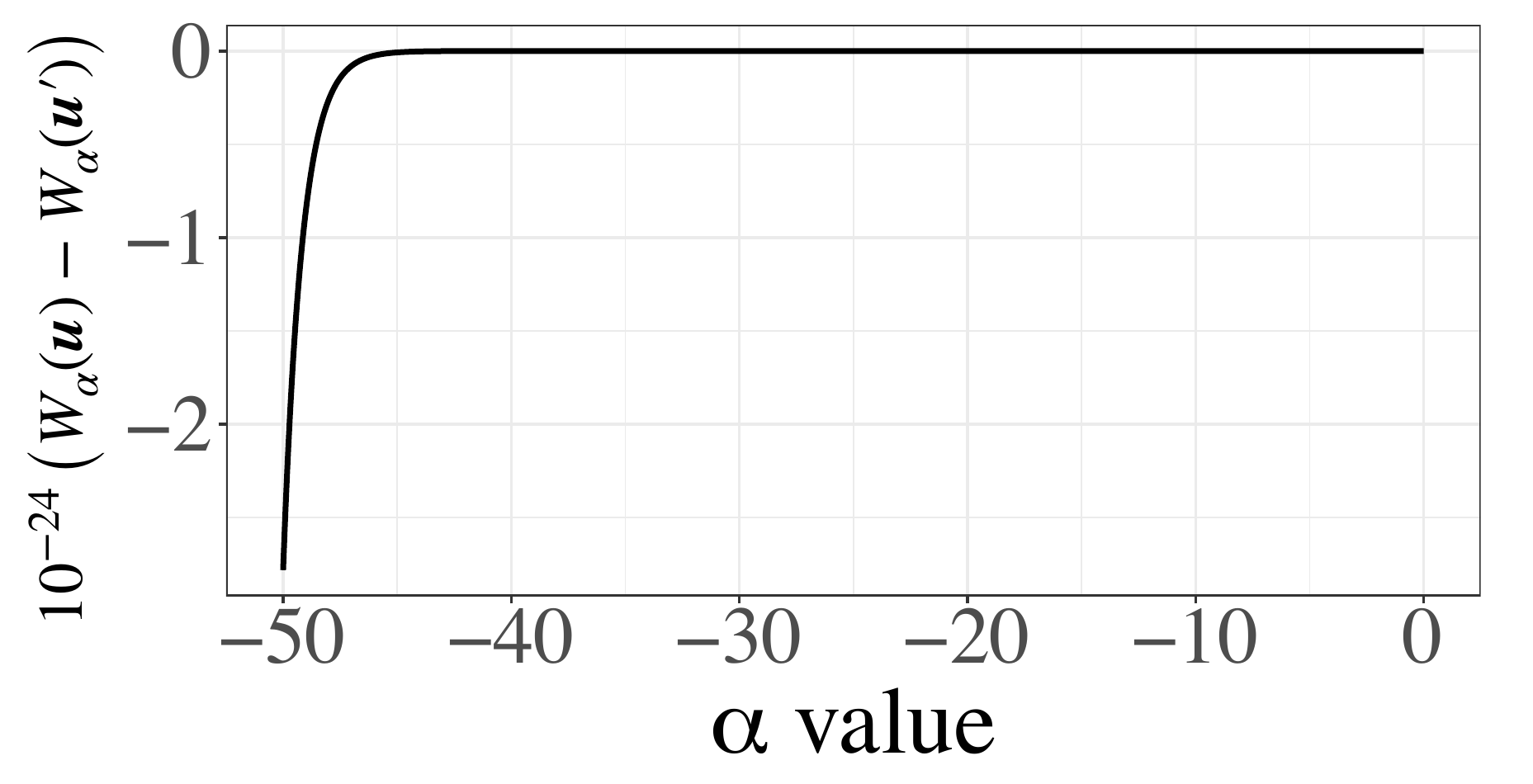}
\end{subfigure}
\caption{The difference of  $\wel_\alpha(\bu)-\wel_\alpha(\bu')$ on the vertical axis versus $\alpha$ on the horizontal axis for 
different welfare functions (this difference is scaled by a factor of $10^{-24}$ on the bottom panel). Top panel: 
$\wel_\alpha(\bu) = \Sigma_{\c\in\C} \n_\c\u_{\c}^\alpha/\alpha $ for $\alpha \in (0,1)$; bottom panel: $\wel_\alpha(\bu) = 
\Sigma_{\c\in\C} \n_\c\u_{\c}^\alpha/\alpha $ for $\alpha < 0$. \label{fig:pop-sep2}}
\end{figure}

We now show that no welfare function that satisfies the first 5 principles will prefer $\bu$ over $\bu'$. Recall that such welfare 
functions are in the form $\wel_\alpha(\bu) = \Sigma_{\c\in\C} \n_\c\u_{\c}^\alpha/\alpha$ for $\alpha < 1$ and $\alpha \neq 0$, 
$\wel_\alpha(\bu) = \Sigma_{\c\in\C} \n_\c\log(\u_{\c})$ for $\alpha=0$.
We verify this claim numerically. In particular Figure~\ref{fig:pop-sep2} plots the 
difference of $\wel_\alpha(\bu)-\wel_\alpha(\bu')$ for $\wel_\alpha(\bu) = \Sigma_{\c\in\C} \n_\c\u_{\c}^\alpha/\alpha$ when $\alpha \in (0,1)$ (top panel) and $\alpha < 0$ (bottom panel). 
This difference is always negative so $\bu'$ is preferred by these welfare functions. For $\wel_\alpha(\bu) = \Sigma_{
\c\in\C} \n_\c\log(\u_{\c}), \wel_\alpha(\bu)-\wel_\alpha(\bu') \approx -4.3$.

We point out that the instance used in the proof (graph structure, probability of spread and the number of seeds) is designed with the sole purpose of simplifying the calculations of the utilities. It is possible to modify this instance to more complicated and realistic instances. 

\end{proof}

\begin{proof}[Proof of Proposition~\ref{pro:disjoint}]
Let $\A$ and $\A'$ denote two neighboring solutions with corresponding utility vectors $\bu = \bu(\A)$ and $\bu' = \bu(\A')$. 
Let $\bu$ denote any of the two utility vectors such that $\Sigma_{\c\in\C} \n_\c\u_{\c} \geq \Sigma_{\c\in\C} \n_\c\u'_{\c}$. 
Without loss of generality, we assume $\bu'$ is sorted in ascending order of the utilities and $\bu$ is permuted so that 
index $\c \in \C$ in both $\bu$ and $\bu'$ corresponds to the same community. This is because we assume that $\wel$ satisfies the 
symmetry principle due to which by permuting a utility vector the value of the welfare function does not change.
Let $\nu$ and $\kappa \in \C$ denote the communities whose utilities are changed between $\bu$ and $\bu'$, i.e., 
we assume $\nu$ and $\kappa$ are the two communities where taking influencer vertices from $\nu$ and giving them to $\kappa$ 
will transfer $\bu'$ into $\bu$. 

To satisfy the condition of the utility gap reduction principle, it should be the case that $\u'_\nu \geq \u'_\kappa$ (i.e., we transfer 
influencer vertices from the group with higher utility to a group with lower utility), otherwise after the transfer from $\bu'$ to $\bu$ the utility 
gap could not get smaller (i.e., $\Delta(\bu) \geq \Delta(\bu')$ in which case the utility gap reduction is not applicable). 

Assuming $\u'_\nu \geq \u'_\kappa$, if $\Delta(\bu) \geq \Delta(\bu')$, again the assumption of the utility gap reduction principle is 
not satisfied, hence the principle is not applicable and there is no need to study this case. Therefore, we further assume 
$\Delta(\bu) < \Delta(\bu')$. We would like to show in this case a welfare function $\wel$ that satisfies all the 5 other principles 
witnesses $\wel(\bu) > \wel(\bu')$.

By assumption 
$\Sigma_{\c\in\C} \n_\c\u_{\c} \geq \Sigma_{\c\in\C} \n_\c\u'_{\c}$. From this, it follows that:
\begin{flalign}
    & \displaystyle \sum_{\c \in \C} \n_\c \left(\u_{\c} - \u'_{\c}\right) \geq 0
    \\
    \Leftrightarrow & \displaystyle \n_\nu \left(\u_{\nu} - 
    \u'_{\nu} \right) + \n_{\kappa} \left(\u_{\kappa} - \u'_{\kappa} \right) \geq 0 \label{eq1}
    \\
    \Leftrightarrow & \displaystyle \sum_{y \in \C : y \leq x} \n_{y} (\u_{y} - \u'_{y}) \geq 0, \; \forall x \in \C, \label{eq2}
\end{flalign}

where both inequalities~\eqref{eq1} and~\eqref{eq2} follow directly from the fact that the utilities of all the other communities 
are the same in both $\bu$ and $\bu'$. Finally, since $\u_{\kappa} > \u'_{\kappa}$ (we are transferring influencer vertices to the community 
$\kappa$), we can apply the influence transfer principle to show that $\wel(\bu) >\wel(\bu')$ as claimed.
\end{proof}

\begin{proof}[Proof of Lemma~\ref{lem:wel-sub-mon}]
As we have shown earlier welfare functions $\wel_\alpha(\bu) = \Sigma_{\c\in\C} \n_\c\log(\u_{\c})$ for $\alpha=0$ and $\wel_\alpha(\bu) = \Sigma_{
\c\in\C} \n_\c\u_{\c}^\alpha/\alpha$ for $\alpha < 1, \alpha \neq 0$ satisfy all the first 5 principles. \citet{lin2011class} show that the composition of a non-decreasing concave function (in our case 
$\log(x), \alpha = 0$ or $x^\alpha/\alpha$ for $\alpha < 1, \alpha \neq 0$) and a non-decreasing submodular function 
(in our case $\u_\c(\A)$) is submodular. Since the sum of submodular functions is submodular, our proposed class of welfare functions is submodular. Our welfare functions also satisfy monotonicity. This is because $\u_\c(\A)$ is monotonically non-decreasing so its composition with another monotonically non-decreasing function  ($\log(x)$ for $\alpha = 0$ or $x^\alpha/\alpha$ for $\alpha < 1, \alpha \neq 0$) will be monotonically non-decreasing. Since our welfare functions are the sum of monotonically non-decreasing function they are also monotone.
\end{proof}
\section{Leximin Fairness and Social Welfare}
\label{sec:leximin}
In this section, we show that leximin fairness can be captured by our welfare maximizing framework. See~\cite{Heidari2018fairness} for more details.
\begin{proposition}
Welfare optimization is equivalent to the leximin fairness, i.e., there exists a constant $\alpha_0$, such for $\alpha < \alpha_0$, an optimal solution to the welfare maximization satisfies leximin fairness and vice versa. 
\end{proposition}

\begin{proof}
Let $\overline\bu = (\overline\u_1,\dots,\overline{\u}_\n) \succcurlyeq \bu(\A) \; \forall \A \in \A^{\star}$, where ``$\succcurlyeq$'' is the lexicographic ordering sign and it indicates that $\overline\bu$ is a leximin fair solution (w.l.o.g. and with a slight abuse of notation, we assume that both $\overline{\bu}$ and $\bu(\A)$ are sorted in increasing order). We aim to show that $\exists \alpha_0 < 0$ such that for any $\alpha 
\leq \alpha_0$,  $\Sigma^{\n}_{i=1} \overline{\u}_{i}^{\alpha}/\alpha \geq \Sigma^{\n}_{i = 1} \u_{i}^{\alpha}(\A)/\alpha, \; \forall \A \in \A^{\star}$. For simplicity we multiply both sides of the inequality by $-1/\alpha$ and since $\alpha <0$ the direction of the inequality sign does not change.

We now prove this inequality by contradiction.
Suppose $\forall \alpha_{1} < 0$,  $\exists\alpha < \alpha_1$ such that $\Sigma_{i = 1}^{\n} -\overline{\u}_{i}^{\alpha} < \Sigma^{\n}_{i=1} - \u_{i}^{\alpha}(\A),  \exists \A \in \A^{\star}$. 
Since $\overline{\bu}$ is a leximin solution then by definition $\overline{\u}_{1} \geq \u_{1}(\A)$.
We consider two cases. First suppose $\overline{\u}_{1} > \u_{1}(\A)$. 
\begin{align*}
    \sum_{i = 1}^{\n} -\overline{\u}_{i}^{\alpha} < \sum^{\n}_{i=1} - \u_{i}^{\alpha}(\A) & \Leftrightarrow  \\
    \frac{\sum_{i = 1}^{\n} -\overline{\u}_{i}^{\alpha}}{\min(\overline{\u}_{1},u_{1}(\A))^{\alpha}} < \frac{\sum^{\n}_{i=1} - \u_{i}^{\alpha}(\A)}{{\min(\overline{\u}_{1},u_{1}(\A))^{\alpha}}} & = \\
    \frac{\sum_{i = 1}^{\n} -\overline{\u}_{i}^{\alpha}}{u_{1}(\A)^{\alpha}} < \frac{\sum^{\n}_{i=1} - \u_{i}^{\alpha}(\A)}{{u_{1}(\A)^{\alpha}}} & \Rightarrow \\
    \lim_{\alpha \rightarrow -\infty} \frac{\sum_{i = 1}^{\n} -\overline{\u}_{i}^{\alpha}}{\u_{1}(\A)^{\alpha}} \leq \lim_{\alpha \rightarrow -\infty} \frac{\sum^{\n}_{i=1} - \u_{i}^{\alpha}(\A)}{{u_{1}(\A)^{\alpha}}} & \Rightarrow \\
    0 \leq -\n_1.
\end{align*}
This is a contradiction since $\n_1 > 0$. Now, suppose $\overline{\u}_{1} = \u_{1}(A)$. In this case, we can eliminate the first terms that involve $\overline{\u}^{\alpha}_{1}$ and $\u^{\alpha}_1$ from the two sides of inequality and redo the above steps iteratively starting from the second biggest element in $\overline{\u}^{\alpha}$. 

Next, we prove the other direction. Let us assume $\overline{\bu}$ is a utility vector such that $\exists \alpha_0 < 0, \forall \alpha 
\leq \alpha_0$,  $\sum^{\n}_{i=1}- \overline{\u}_{i}^{\alpha} \geq \sum^{\n}_{i = 1} -\u_{i}^{\alpha}(\A), \; \forall \A \in \A^{\star}$. W.l.o.g, we can assume that $\overline{\u}_1 \neq \u_{1}$ otherwise we can remove those terms that are equal and the proof still holds. However, we assume this for ease of exposition. It follows that
$$
\frac{\sum^{\n}_{i=1}- \overline{\u}_{i}^{\alpha}}{\min(\overline{\u}_{1},\u_{1})^{\alpha}} \geq \frac{\sum^{\n}_{i = 1} -\u_{i}^{\alpha}(\A)}{\min(\overline{\u}_{1},\u_{1})^{\alpha}}, \; \forall \A \in \A^{\star}. 
$$
If $\min(\overline{\u}_{1}, \u_{1}) = \overline{\u}_{1}$ meaning that $\u_{1} > \overline{\u}_{1}$ we have $-C - \epsilon(\alpha) \geq -\delta(\alpha, \A), \; \forall \A \in \A^{\star}$ where $C>0$ is a constant (equal to the number of entities in $\overline{\bu}$ that are equal to $\overline{\u}_{1}$) and both $\epsilon\geq 0$ and $\delta\geq 0$ are functions of $\alpha$ and can be made arbitrarily small by decreasing $\alpha$.  This is a contradiction which means that $\min(\overline{u}_{1}, \u_{1}) = \u_{1}$, i.e., $\overline{\u}_{1} \geq {\u}_{1}$. By continuing this procedure, we can establish that $\overline{\bu} \succcurlyeq \bu$. This completes the proof.
\end{proof}
\section{Omitted Proofs from Table~\ref{tab:welfare-exisitng-notions}}
\label{sec:table}
In this section we provide detailed description of the entries of Table~\ref{tab:welfare-exisitng-notions} and their derivations. 

\subsection{Monotonicity}
\begin{proposition}
\label{pro:mono-dp-exact}
Exact DP does not satisfy monotonicity.
\end{proposition}
\begin{proof}
\begin{figure}[ht!]
\centering
\begin{tikzpicture}
[scale=0.45, red node/.style={circle,fill=gray, draw=black}, blue node/.style = {rectangle, fill = gray, draw=black}
, scale=1.0, every
edge/.style={->,-> = latex'}
]
\tikzset{scale=0.9, edge/.style = {-,> = latex'}}
\tikzset{scale=0.9, directededge/.style = {->,> = latex'}}

\node[blue node] (0) at  (-5, 0){};
\node[blue node] (1) at  (-3, 0){};
\node[blue node] (2) at  (-3.59, 1.41){};
\node[blue node] (3) at  (-5, 2){};
\node[blue node] (4) at  (-6.41, 1.41){};
\node[blue node] (5) at  (-7, 0){};
\node[blue node] (6) at  (0, -1.5){};
\node[blue node] (7) at  (0,0){};
\node[blue node] (8) at  (0, 1.5){};
\draw[edge] (0) to (1);
\draw[edge] (0) to (2);
\draw[edge] (0) to (3);
\draw[edge] (0) to (4);
\draw[edge] (0) to (5);

\node[red node] (10) at  (5, 0){};
\node[red node] (11) at  (3, 0){};
\node[red node] (12) at  (3.59, 1.41){};
\node[red node] (13) at  (5, 2){};
\node[red node] (14) at  (6.41, 1.41){};
\node[red node] (15) at  (7, 0){};
\node[red node] (16) at  (6.41, -1.41){};
\node[red node] (17) at  (5,-2){};
\node[red node] (18) at  (3.59, -1.41){};
\draw[edge] (10) to (11);
\draw[edge] (10) to (12);
\draw[edge] (10) to (13);
\draw[edge] (10) to (14);
\draw[edge] (10) to (15);
\draw[edge] (10) to (16);
\draw[edge] (10) to (17);
\draw[edge] (10) to (18);
\end{tikzpicture}
\caption{Companion figure to Proposition~\ref{pro:mono-dp-exact}. The network consists of two communities circle and square, each of size $\n$. }
\label{fig:exact-dp-level-down}
\end{figure}
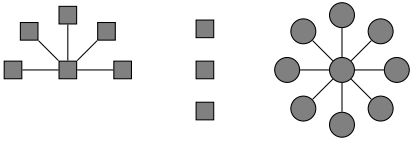
Let $\k=2$ and $\p \in (0,1)$. 
Consider a graph $\G$ as shown in Figure~\ref{fig:exact-dp-level-down} consisting of two communities, square and circle, each of size $\n$ (for large enough $\n$). 
The circle community consists of a star network of size $\n$. The square community contains a star network of size $2+\p(\n-2)$ and $(\n-2)(1-\p)$ singletons. Consider two solutions $\A$ and $\A'$. $\A$ will select a seed from the periphery of the star for the circle community and allocate the other seed to the center of the star for the square community. $\A'$ on the other hand allocates each of the seeds to the center of the stars. Let $\bu = \bu(\A)$ and $\bu'=\bu(\A')$ denote the corresponding allocations of $\A$ and $\A'$. The utility vectors for these allocations are $\bu=((1+p+p^2(\n-2))/\n,(1+p+p^2(\n-2))/\n)$ and $\bu'=((1+p(\n-1))/\n, (1+p+p^2(\n-2))/\n)$, respectively. Clearly, $\bu < \bu'$. So by monotonicity $\bu'$ is preferred to $\bu$. However, only $\bu$ satisfies the exact DP. Hence, DP does not satisfy monotonicity. 
\end{proof}

\begin{proposition}
\label{pro:mono-dp}
Approximate DP does not satisfy monotonicity. \label{lem:eqopp-levelingdown}
\end{proposition}
\begin{proof}
Consider a graph $\G$ as shown in Figure~\ref{fig:eqopp-levelingdown} consisting of two communities, square and circle, each of size $\n$. 
We choose an arbitrary $\delta \in (0,1),$ to reflect the arbitrary strictness of a decision maker. Let $\delta < \p < \sqrt\delta$, $\k=2$ and $
\n > \max\left(3\p/(\p-\delta),1/(\delta-\p^2)\right)$. The optimal solution $\A$ of the influence maximization problem 
chooses the center of the star and any disconnected square vertex. In $\A$, the utility of circle and square communities are 
$\left(1+(\n-1)\p\right)/\n$ and $(1+2\p)/\n$, respectively and the utility gap exceeds $\delta$ (so this solution does not satisfy the DP constraints).  By imposing DP, any fair solution is to choose one vertex from the periphery of the circle community and 
one from the isolated square vertices. For a fair solution $\A'$, the utilities of circle and square are $\left(1+\p+\p^2(\n-2)\right)/\n$ and $(1+2\p^2)/\n$, respectively. 
Given the range of $\n$, the utility gap is less than $\delta$ so approximate DP is satisfied. Since the utility of both communities have degraded, any monotone welfare function will prefer $\A'$ (and its corresponding utility vector) over $\A$. However, only $\A'$ is DP fair and hence it is preferred over $\A$ by DP. We point out that the graph used in the proof is directed. This is for ease of exposition. It is possible to create a more complex example with an undirected graph.
\end{proof}

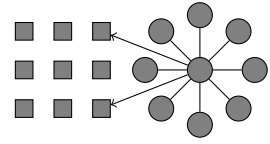
\begin{figure}[ht!]
\centering
\begin{tikzpicture}
[scale=0.45, red node/.style={circle,fill=gray, draw=black}, blue node/.style = {rectangle, fill = gray, draw=black}
, scale=1.0, every
edge/.style={->,-> = latex'}
]
\tikzset{scale=0.9, edge/.style = {-,> = latex'}}
\tikzset{scale=0.9, directededge/.style = {->,> = latex'}}

\node[blue node] (1) at  (1.41, 0){};
\node[blue node] (2) at  (-1.41, 0){};
\node[blue node] (3) at  (0, 0){};
\node[blue node] (4) at  (-1.41, 1.41){};
\node[blue node] (5) at  (0, 1.41){};
\node[blue node] (6) at  (1.41, 1.41){};
\node[blue node] (7) at  (-1.41, -1.41){};
\node[blue node] (8) at  (0,-1.41){};
\node[blue node] (9) at  (1.41, -1.41){};
\node[red node] (10) at  (5, 0){};
\node[red node] (11) at  (3, 0){};
\node[red node] (12) at  (3.59, 1.41){};
\node[red node] (13) at  (5, 2){};
\node[red node] (14) at  (6.41, 1.41){};
\node[red node] (15) at  (7, 0){};
\node[red node] (16) at  (6.41, -1.41){};
\node[red node] (17) at  (5,-2){};
\node[red node] (18) at  (3.59, -1.41){};
\draw[edge] (10) to (11);
\draw[edge] (10) to (12);
\draw[edge] (10) to (13);
\draw[edge] (10) to (14);
\draw[edge] (10) to (15);
\draw[edge] (10) to (16);
\draw[edge] (10) to (17);
\draw[edge] (10) to (18);
\draw[->] (10) to (6);
\draw[->] (10) to (9);
\end{tikzpicture}
\caption{Companion figure to Proposition~\ref{lem:eqopp-levelingdown}. The network consists of two communities circle and square, each of size $\n$. All edges except the two 
shown by arrows are undirected meaning that influence can spread both ways.}
\label{fig:eqopp-levelingdown}
\end{figure}

\begin{proposition}
\label{pro:mono-others}
Consider a general fairness notion as a set of constraints
$\mathcal F = \left\{ \bm u\in[0,1]^\NC : \uc\geq l_\c, \;\; \forall \c \in \C \right\}$ where $l_\c \; \forall \c\in\C$ are arbitrary lower-bound values. The considered fairness notion satisfies the monotonicity principle.
\label{lem:general-leveldown}
\end{proposition}
\begin{proof}
Let $\A$ and $\A'\in \A^\star$ denote two solutions whose corresponding utility vectors $\bu = \bu(\A)$ and $\bu'= \bu(\A)$ are feasible ($\bu,\bu'\in\mathcal F$) such that $\bu < \bu'$. 
Given the objective function of the influence maximization is equivalent to $\Sigma_{\c \in\C} \n_\c\uc(\A)$ and that all $\n_\c$ values are positive the objective values of $\bu'$ is strictly better than $\bu$. Hence, $\wel(\bu) < \wel(\bu')$ and the monotonicity is satisfied.
\end{proof}

As we have shown in Section~\ref{sec:existing-fairness}, both maximin and DC can be written as constraints that are compatible with the fairness definition in Proposition~\ref{pro:mono-others}. The utilitarian solution corresponds to setting all the lower bounds to 0. Hence all of them satisfy monotonicity.
\begin{corollary}
\label{cor:mono-dc-maximin}
DC, MMF and utilitarian satisfy monotonicity. 
\end{corollary}

\subsection{Symmetry}
It is straightforward to show that DP (exact and approximate), maximin and utilitarian fairness satisfy the symmetry principle. DC, however, does not satisfy the symmetry principle. Based on its definition, DC can place different lower-bounds on the utility of different communities. Hence, by permuting a utility vector we may no longer be able to satisfy the DC constraints (see Definition~\ref{eq:dc}).

\subsection{Independence of Unconcerned Individuals}
\begin{proposition}
\label{pro:dp-iuc}
Exact and approximate DP do not satisfy the independence of unconcerned individuals. 
\end{proposition}
\begin{proof}
Consider two utility vectors: $\bu = ((1+3\delta)/8, (1-\delta)/8)$ and $\bu' = ((1+\delta)/4, (1-\delta)/8)$ for $\delta\in[0,1)$. Both exact and approximate DP strictly prefer $\bu$ over $\bu'$. Let us substitute the second component of both vectors by  $(1+\delta)/4$. Therefore, we obtain $\bm v = \bu|^2 (1+\delta)/4  = ((1+3\delta)/8, (1+\delta)/4)$ and $\bm v'=\bu'|^2 (1+\delta)/4 = ((1+\delta)/4, (1+\delta)/4)$. In contrast to the previous case, both approximate and exact DP prefer $\bm v'$ over $\bm v$. Note that while the construction does does not involve an instance of the influence maximization problem, it is possible to provide such an instance to witness the claim as follows. 

\begin{figure}[ht!]
\centering
\begin{tikzpicture}
[scale=0.45, red node/.style={circle,fill=gray, draw=black}, blue node/.style = {rectangle, fill = gray, draw=black}
, scale=1.0, every
edge/.style={->,-> = latex'}
]
\tikzset{scale=0.9, edge/.style = {-,> = latex'}}
\tikzset{scale=0.9, directededge/.style = {->,> = latex'}}
\node[red node] (101) at  (1.75, -1){};
\node[red node] (101) at  (1.75, 1){};
\node[blue node] (201) at  (14, 2){};
\node[blue node] (202) at  (14, 1){};
\node[blue node] (203) at  (14, 0){};
\node[blue node] (204) at  (14, -1){};
\node[blue node] (205) at  (14, -2){};
\node[red node] (1) at  (10.5, 0){};
\node[red node] (2) at  (8.5, 0){};
\node[red node] (3) at  (9.09, 1.41){};
\node[red node] (4) at  (10.5, 2){};
\node[red node] (5) at  (11.91, 1.41){};
\node[red node] (6) at  (12.5, 0){};
\node[red node] (7) at  (11.91, -1.41){};
\draw[edge] (1) to (2);
\draw[edge] (1) to (3);
\draw[edge] (1) to (4);
\draw[edge] (1) to (5);
\draw[edge] (1) to (6);
\draw[edge] (1) to (7);
\node[blue node] (41) at  (17, 0){};
\node[blue node] (42) at  (15, 0){};
\node[blue node] (43) at  (15.59, 1.41){};
\node[blue node] (44) at  (17, 2){};
\node[blue node] (45) at  (18.41, 1.41){};
\node[blue node] (46) at  (19, 0){};
\node[blue node] (47) at  (18.41, -1.41){};
\draw[edge] (41) to (42);
\draw[edge] (41) to (43);
\draw[edge] (41) to (44);
\draw[edge] (41) to (45);
\draw[edge] (41) to (46);
\draw[edge] (41) to (47);
\node[blue node] (51) at  (22.5, 0){};
\node[blue node] (52) at  (20.5, 0){};
\node[blue node] (53) at  (21.09, 1.41){};
\node[blue node] (54) at  (22.5, 2){};
\node[blue node] (55) at  (23.91, 1.41){};
\node[blue node] (56) at  (24.5, 0){};
\draw[edge] (51) to (52);
\draw[edge] (51) to (53);
\draw[edge] (51) to (54);
\draw[edge] (51) to (55);
\draw[edge] (51) to (56);
\node[red node] (10) at  (5, 0){};
\node[red node] (11) at  (3, 0){};
\node[red node] (12) at  (3.59, 1.41){};
\node[red node] (13) at  (5, 2){};
\node[red node] (14) at  (6.41, 1.41){};
\node[red node] (15) at  (7, 0){};
\node[red node] (16) at  (6.41, -1.41){};
\node[red node] (17) at  (5,-2){};
\node[red node] (18) at  (3.59, -1.41){};
\draw[edge] (10) to (11);
\draw[edge] (10) to (12);
\draw[edge] (10) to (13);
\draw[edge] (10) to (14);
\draw[edge] (10) to (15);
\draw[edge] (10) to (16);
\draw[edge] (10) to (17);
\draw[edge] (10) to (18);
\end{tikzpicture}
\caption{Companion figure to Proposition~\ref{pro:dp-iuc}. The network consists of two communities circle and square each of size $\n$.}
\label{fig:dp-iuc}
\end{figure}
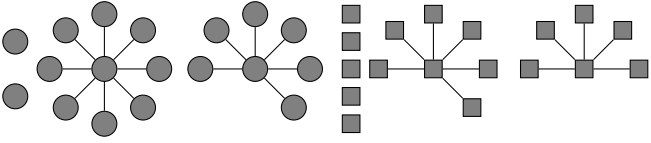

Figure~\ref{fig:dp-iuc} demonstrates the instance witnessing the claim. We consider an influence maximization problem with two communities: circle (first component of the utility vector) and square (second component of the utility vector), each of size $\n$ (for $\n$ large enough). We assume $\p=1$ and $\k=2$. The circle community consists of three components: two star components of size $\n(1+\delta)/4$ (small) and $\n(1+3\delta)/8$ (large) and $5\n(1-\delta)/8$ isolated vertices. The square community consists of three components as well: two star components of size $\n(1-\delta)/8$ (small) and $\n(1+\delta)/4$ (large) and $\n(5-\delta)/8$ isolated vertices. Solution $\bu$ ($\bu'$) corresponds to selecting a seed vertex from the large (small) star component of the circle community and a vertex from the small star component of the square community. Allocation $\bm v$ ($\bm v'$) corresponds to selecting a seed vertex from the large (small) star component of the circle community a vertex from the large star component of the square community. Note that the choice of $\p=1$ is merely for the ease of exposition and the example network can be modified to accommodate $\p < 1$.
\end{proof}

Henceforth, we only discuss utility vectors when appropriate. In all such cases, there exist instances of the influence maximization problem which witness these utility vectors. We have demonstrated one such instance in the proof of Proposition~\ref{pro:dp-iuc}, but we omit the details from the remaining proofs for simplicity. 

\begin{proposition}
\label{pro:dc-iuc}
DC does not satisfy the independence of unconcerned individuals.
\end{proposition}
\begin{proof}
Consider an instance of the influence maximization problem with 2 communities were the lower bound set by DC for both communities is $0.4$. Also consider 2 solutions with corresponding utility vectors $\bu = (0.5, 0.5)$ and $\bu' = (0.5, 0.3)$. Therefore, only $\bu$ satisfies DC and hence DC prefers $\bu$ over $\bu'$. Let us substitute the first component with $0.35$. Therefore, we obtain $\bm v = \bu|^1 0.35 = (0.35, 0.5)$ and $\bm v'= \bu'|^1 0.35 = (0.35, 0.3)$. In contrast to the previous case, both solutions are infeasible with respect to the DC (hence -$\infty$ welfare, see Section~\ref{sec:fair-connection}). Therefore, while  $\wel(\bu) > \wel(\bu')$, it does not hold that $\wel(\bm v) > \wel(\bm v')$. 
\end{proof}

\begin{proposition}
\label{pro:lexi-iuc}
MMF does not satisfy the independence of unconcerned individuals.
\end{proposition}
\begin{proof}
Consider an instance of the influence maximization problem with 3 communities of equal size.
Also consider 2 solutions with corresponding utility vectors $\bu = (0.3, 0.6, 0.4)$ and $\bu' = (0.3, 0.2, 0.8)$. Maximin fairness strictly prefers $\bu$ over $\bu'$. Let us substitute the first component with $0.1$. Therefore, we obtain $\bm v = \bu|^1 0.1 = (0.1, 0.6, 0.4)$ and $\bm v'= \bu'|^1 0.1 = (0.1, 0.2, 0.8)$. Maximin fairness is indifferent between $\bm v$ and $\bm v'$ (both have the same worst-case utility and total utility) which shows that maximin fairness does not satisfy the independence of unconcerned individuals.
\end{proof}

Note that the utilitarian satisfies the independence of unconcerned individuals because if $\wel(\bu) = \Sigma_{\c \in\C} \n_\c\uc > \Sigma_{\c \in\C} \n_\c\uc' = \wel(\bu')$ then $\wel(\bu|^{\c}b') < \wel(\bu'|^{\c}b')$ since $\n_\c$, $\u_c$, $\uc'$ and $b'$ are all non-negative.

\subsection{Affine Invariance}
Exact DP satisfies affine invariance principle because a linear transformation over a uniform vector will remain uniform. However, for approximate DP this is not the case. More particularly, for any utility vector $\bu$ that is $\delta$-DP for $\delta \in (0,1)$ and an affine transformation of the form $\bu' = \alpha \bu + \beta$, $\bu'$ satisfies $\alpha\delta$-DP. Therefore, for $\alpha > 1$, $\bu'$ does not satisfy $\delta$-DP. Similarly, DC does not satisfy this principle either. This is because after the transformation the constraints may not be satisfied (e.g., when $\alpha < 1/\min_{\c\in\C} U_c $). It is known that MMF satisfies this principle~\cite{BertsimasFT11-pricefair}. The same holds for the utilitarian objective because if $\Sigma_{\c \in \C}\n_\c \uc > \Sigma_{\c \in \C}\n_
\c \u'_\c$, then $\Sigma_{\c \in \C}\n_\c \alpha\uc+\beta > \Sigma_{\c \in \C}\n_
\c \alpha\u'_\c+\beta$ since $\alpha > 0$.

\subsection{Influence Transfer Principle. }

\begin{proposition}
\label{pro:dp-trans-2}
Exact and approximate DP do not satisfy the influence transfer principle.
\end{proposition}
\begin{proof}
Let $\delta\in[0,1)$ denote the parameter of DP.
Consider utility vectors $\bu = ((1+\delta)/2, 0)$ and  $\bu' = (\delta, 0)$. The sizes of the all the communities are the same. Based on the influence transfer principle, $\wel(\bu) > \wel(\bu')$, however, DP strictly prefers $\bu'$ over $\bu$.  
\end{proof}

\begin{proposition}
\label{pro:dc-trans}
DC does not satisfy the influence transfer principle.
\end{proposition}
\begin{proof}
Let $\bu = (0.5, 0.5)$ and $\bu' = (0.3, 0.6)$ denote the utility vectors of two allocations where sizes of the all the communities are the same. Suppose the lower bounds set by DC are $0.25$ and  $0.55$, respectively. This means that only $\bu'$ satisfies DC.  Based on the transfer principle, $\wel(\bu) > \wel(\bu')$, however, $\bu$ does not satisfy DC and DC strictly prefers $\bu'$ over $\bu$. 
\end{proof}
\begin{proposition}
\label{pro:lexi-trans}
MMF does not satisfy the influence transfer principle. 
\end{proposition}
\begin{proof}
Consider two utility vectors $\bu = (0.2, 0.4, 0.6)$ and $\bu'=(0.2, 0.2, 0.8)$. The sizes of the all the communities are the same. A fairness notion satisfying the influence transfer principle strictly prefers $\bu$ over $\bu'$. However, Maximin is indifferent between $\bu$ and $\bu'$ as they both obtain the same worst-case utility. 
\end{proof}

\begin{proposition}
\label{pro:util-trans}
Utilitarian does not satisfy the influence transfer principle. 
\end{proposition}
\begin{proof}
Consider two utility vectors $\bu = (0.5, 0.5)$ and $\bu' = (0.3, 0.7)$. The sizes of the all the communities are the same. Based on the transfer principle, $\wel(\bu) > \wel(\bu')$, however, the utilitarian approach is indifferent between $\bu$ and $\bu'$ as both solutions lead to the same total utility.
\end{proof}

\subsection{Utility Gap Reduction. }
\begin{proposition}
\label{pro:dp-ugr}
\emph{DP} satisfies the utility gap reduction if and only if $\delta = 0$.
\label{lem:eqopp-richgetricher}
\end{proposition}
\begin{proof}
It is easy to show that if $\delta = 0$, DP satisfies the utility gap reduction principle. We can prove this by contradiction. Suppose that $\delta = 0$ and DP does not satisfy the utility gap reduction principle. From this, it follows that given two utility vectors $\bu, \bu'$ such that $\sum_{\c \in \C} \n_c \u_\c \geq \sum_{\c \in \C} \n_c \u'_\c$ if $\Delta \bu < \Delta \bu'$, DP can strictly prefer $\bu'$. If $\bu'$ is preferred, then $\bu'$ is feasible and it must be that $\Delta \bu' = 0$ and $\Delta \bu < 0$ which is not possible. Next, we show that if $\delta \neq 0$, DP does not satisfy this principle over all instances of the influence maximization problem. 

For the proof, we use the example used in the proof of Proposition~\ref{pro:gap}. In that setting there are two solutions with utility gap $0.52$ and $0.5$ with the same total utility. First, let us assume $\delta > 0.02$. In this case, since both solutions satisfy DP constraints, they are both feasible and the total utility of both solutions is equal (= 180), however, DP does not \emph{strictly} prefer the solution with smaller gap. In fact, both solutions are feasible with the same objective value and DP does not favor one solution to the other. For $\delta \leq 0.02$, we can use the same example graph (see Figure~\ref{fig:pop-sep}) and add enough isolated vertices to each community until the gap between the solutions becomes small enough to pass the $\delta$ threshold. 
\end{proof}

\begin{figure}[ht!]
\centering
\begin{tikzpicture}
[scale=0.55, red node/.style={circle,fill=gray, draw=black}, blue node/.style = {rectangle, fill = gray, draw=black}
, scale=1.0, every
edge/.style={->,-> = latex'}
]
\tikzset{scale=0.8, edge/.style = {-,> = latex'}}
\tikzset{scale=0.8, directededge/.style = {->,> = latex'}}

\node[blue node] (1) at  (1.41, 0){};
\node[blue node] (6) at  (1.41, 1.41){};
\node[blue node] (9) at  (1.41, -1.41){};
\node[red node] (10) at  (5, 0){};
\node[red node] (11) at  (3, 0){};
\node[red node] (12) at  (3.59, 1.41){};
\node[red node] (13) at  (5, 2){};
\node[red node] (14) at  (6.41, 1.41){};
\node[red node] (15) at  (7, 0){};
\node[red node] (16) at  (6.41, -1.41){};
\node[red node] (17) at  (5,-2){};
\node[red node] (18) at  (3.59, -1.41){};
\draw[edge] (10) to (11);
\draw[edge] (10) to (12);
\draw[edge] (10) to (13);
\draw[edge] (14) to (10);
\draw[edge] (15) to (10);
\draw[edge] (16) to (10);
\draw[edge] (17) to (10);
\draw[edge] (18) to (10);
\end{tikzpicture}
\caption{Companion figure to Proposition~\ref{prop:dc-population-separation}
of a graph with two communities: $\n$ black vertices and $\n/3$ white vertices
for $\n=9$. We choose $\k=4$ and arbitrary $\p < 1$. All edges are undirected, meaning that influence can spread both ways.}
\label{fig:dc-levelingdown}
\end{figure}
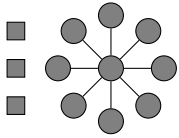

\begin{proposition}
\label{pro:dc-ugr}
\emph{DC} does not satisfy utility gap reduction principle.
\label{prop:dc-population-separation}
\end{proposition}
\begin{proof}
Consider the network $\mathcal G$ as in Figure~\ref{fig:dc-levelingdown} consisting of two communities white and black with size $\n/3$ and $\n$, respectively. Suppose $\k = 4$ and $\p < 1$. Without DC, an optimal solution places one seed vertex at the center of the black group and allocates the remaining 3 vertices to the white group. We let $\bm u$ denote this solution. Thus, the utility of the black and white groups will be equal to $\left(1+(\n-1)\p\right)/\n$ and $9/\n$. Since there is no edge between the black and white communities, DC (See Definition~\ref{eq:dc}) reduces to how to optimally choose one seed vertex from white and the remaining 3 from the black group. After imposing DC, the utility of the black and white groups will be equal to $(3+(\n-3)\p)/\n$ and $3/\n$. We let $\bm w$ denote one such solution that satisfies DC. While $\bm u$ has a higher total utility and a smaller utility gap, DC strictly prefers $\bm w$ with higher utility gap and lower total utility. Therefore, DC does not satisfy utility gap reduction principle.
\end{proof}

\begin{proposition}
\label{pro:lexi-ugr}
MMF does not satisfy the utility gap reduction 
\label{lem:maximin-rich}
\end{proposition}

\begin{proof}
We prove the statement via the example in Figure~\ref{fig:maximin-richgetricher} which depicts a network with three groups: blue, black and white. We fix $\k = 1$ and $\p > 3/4$. 
The graph corresponds to the case where $\p=1$ but the example will hold for arbitrary $\p$ by setting the number of isolated green vertices to be
$\lceil 21/\p\rceil$. 
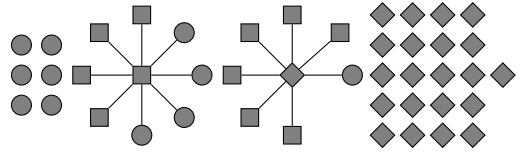
\begin{figure}[ht!]
\centering
\begin{tikzpicture}
[scale=0.5, red node/.style={circle,fill=gray, draw=black, scale = 0.8}, blue node/.style = {rectangle, fill = gray, draw=black}
, scale=1.0, every
edge/.style={->,> = latex'}
]
\tikzset{scale=0.8, edge/.style = {-,> = latex'}}
\tikzset{grey node/.style = {diamond, fill = gray, draw=black, scale=0.7}}

\node[blue node] (1) at  (0, 0){};
\node[red node] (6) at  (2, 0){};
\node[blue node] (2) at  (-2, 0){};
\node[blue node] (4) at  (0, 2){};
\node[red node] (8) at  (0,-2){};
\node[blue node] (3) at  (-1.41, 1.41){};
\node[red node] (5) at  (1.41, 1.41){};
\node[red node] (7) at  (1.41, -1.41){};
\node[blue node] (9) at  (-1.41, -1.41){};

\node[grey node] (11) at  (5, 0){};
\node[red node] (16) at  (7, 0){};
\node[blue node] (12) at  (3, 0){};
\node[blue node] (14) at  (5, 2){};
\node[blue node] (18) at  (5,-2){};
\node[blue node] (13) at  (3.59, 1.41){};
\node[blue node] (15) at  (6.59, 1.41){};
\node[blue node] (19) at  (3.59, -1.41){};

\node[grey node] (20) at (8,0){};
\node[grey node] (21) at (8,1){};
\node[grey node] (22) at (8,2){};
\node[grey node] (23) at (8,-1){};
\node[grey node] (24) at (8,-2){};
\node[grey node] (25) at (9,0){};
\node[grey node] (26) at (9,1){};
\node[grey node] (27) at (9,2){};
\node[grey node] (28) at (9,-1){};
\node[grey node] (29) at (9,-2){};
\node[grey node] (40) at (10,0){};
\node[grey node] (41) at (10,1){};
\node[grey node] (42) at (10,2){};
\node[grey node] (43) at (10,-1){};
\node[grey node] (44) at (10,-2){};
\node[grey node] (45) at (11,0){};
\node[grey node] (46) at (11,1){};
\node[grey node] (47) at (11,2){};
\node[grey node] (48) at (11,-1){};
\node[grey node] (49) at (11,-2){};
\node[grey node] (50) at (12,0){};

\node[red node] (30) at (-3,0){};
\node[red node] (31) at (-3,1){};
\node[red node] (33) at (-3,-1){};
\node[red node] (34) at (-4,-0){};
\node[red node] (35) at (-4,1){};
\node[red node] (36) at (-4,-1){};

\draw[edge] (2) to (1);
\draw[edge] (1) to (3);
\draw[edge] (1) to (4);
\draw[edge] (1) to (5);
\draw[edge] (1) to (6);
\draw[edge] (1) to (7);
\draw[edge] (1) to (8);
\draw[edge] (1) to (9);

\draw[edge] (12) to (11);
\draw[edge] (11) to (13);
\draw[edge] (11) to (14);
\draw[edge] (11) to (15);
\draw[edge] (11) to (16);
\draw[edge] (11) to (18);
\draw[edge] (11) to (19);
\end{tikzpicture}
\caption{Companion figure to Proposition~\ref{lem:maximin-rich} for the case of $p=1$.
The network consists of three groups: white, blue and black. The edges are undirected so the influence can spread both ways. For arbitrary $p$, the number of isolated black vertices should scale to $\lceil 21/\p\rceil$.}
\label{fig:maximin-richgetricher}
\end{figure}

Consider one solution that targets the center of the bigger star component. Thus, the utilities of blue, black and white will be $(1+4\p)/11$, $4\p/11$ and $0$, respectively. This results in utility gap equal to $(1+4\p)/11$. By imposing leximin, the optimal fair solution selects the center of the smaller star component and the optimal fair utilities of blue, black and white will be $6\p/11$, $\p/11$ and $1/(1+\lceil 21/\p\rceil)$, respectively and we observe a utility gap~$
{6\p}/{11}-{1}/({1+\lceil {21}/{\p}\rceil}) \geq {6\p}/{11} - {1}/{22} > ({1+4\p})/{11},
$ where we used $\p>3/4$ in the last inequality. In conclusion,  while the first solution has a higher total utility (= 9) and lower utility gap compared to the second solution (total utility = 8), leximin still strictly prefers the second solution. This concludes the proof.
\end{proof}

\begin{proposition}
\label{pro:util-ugr}
Utilitarian does not satisfy the utility gap reduction 
\end{proposition}
\begin{proof}
Consider an instance of the influence maximization problem where all the communities are of size 10. Let $\bu = (0.5, 0.5, 0.5)$ and $\bu' = (0.2, 0.5, 0.8)$. Both $\bu $ and $\bu'$ achieve the same total utility (=15). Thus, the utilitarian approach is indifferent between $\bu$ and $\bu'$. According to the utility gap reduction $\bu$ is strictly preferred. We note that in this special instance, both influence transfer principle and utility gap reduction apply and according to both principles $\bu$ is strictly preferred to $\bu'.$
\end{proof}

\section{Omitted Details from Section~\ref{sec:exp}}
\label{sec:omitted-exp}

\subsection{Estimating the SBM Parameters for Landslide Risk Management Application}
\label{sec:exp-sbm-estimate}
In order to qualitatively and formatively describe the network structure, the research team conducted several in-person semi-structured interviews in Sitka, Alaska from 2018-2020. These interviews were conducted with individuals who were identified as “community leaders” or ``popular community members'' through word-of-mouth, and then subsequently through respondent-driven sampling, a broader range of community members were interviewed (n=14). In these semi-structured interviews, respondents were asked to 1) sort and describe community groups and 2) identify “cliques” and ``isolates'' as they relate to an early landslide warning system. The former resulted in developing, to the extent possible, discrete a priori community groups. The latter helped to inform the relationships between and within these groups. The interviewer took notes which listed the responses and through a tallying and pile sorting exercise, attempted to seek consensus in definitions of community groups.
The formative research resulted in cliques based on occupation, political affiliation, age, and local recreational activities. Many cliques were overlapping with shared attributes (e.g. people from two different occupations share a political affiliation and frequent the same local pub), however for the purposes of this formative exercise, these community groups were qualitatively coerced into discrete classifications.  These resulted in 16 community groups that include political affiliation, time spent in Sitka (e.g. new arrival or tourist vs. long-term resident), occupation, and whether or not a parent of a child in the public-school system. The community size estimates were developed based on a 2018 Sitka Economic Development Survey, particularly for the occupation-based community groups, as well as publicly available voter records for political affiliations. Several attributes, namely age, specific occupation, time spent in Sitka, and parental status were unavailable in existing datasets, and therefore required the use of proxies and assumptions for estimating community group sizes.
Once the community group sizes were estimated, based on the formative research notes on social cohesion, cliques, and isolates, we further developed assumptions on within and between-community connectedness. For example, if a respondent suggested that there may be very close relationships between two cliques, we assumed a higher relative p(b) than between two cliques which had less similar attributes. For simplicity, we limited the absolute probabilities for withing-community and between-community probabilities between 0.00 and 0.10. We then sense-checked these absolute probabilities with several of the initial formative research respondents. These absolute probabilities were then organized into a $16\times 16$ adjacency matrix to facilitate simulations for influence maximization.

\subsection{Relative Community Sizes}
\label{sec:exp-relative}
\begin{figure}[ht!]
\centering
\begin{subfigure}
\centering
\includegraphics[width = 0.33\textwidth]{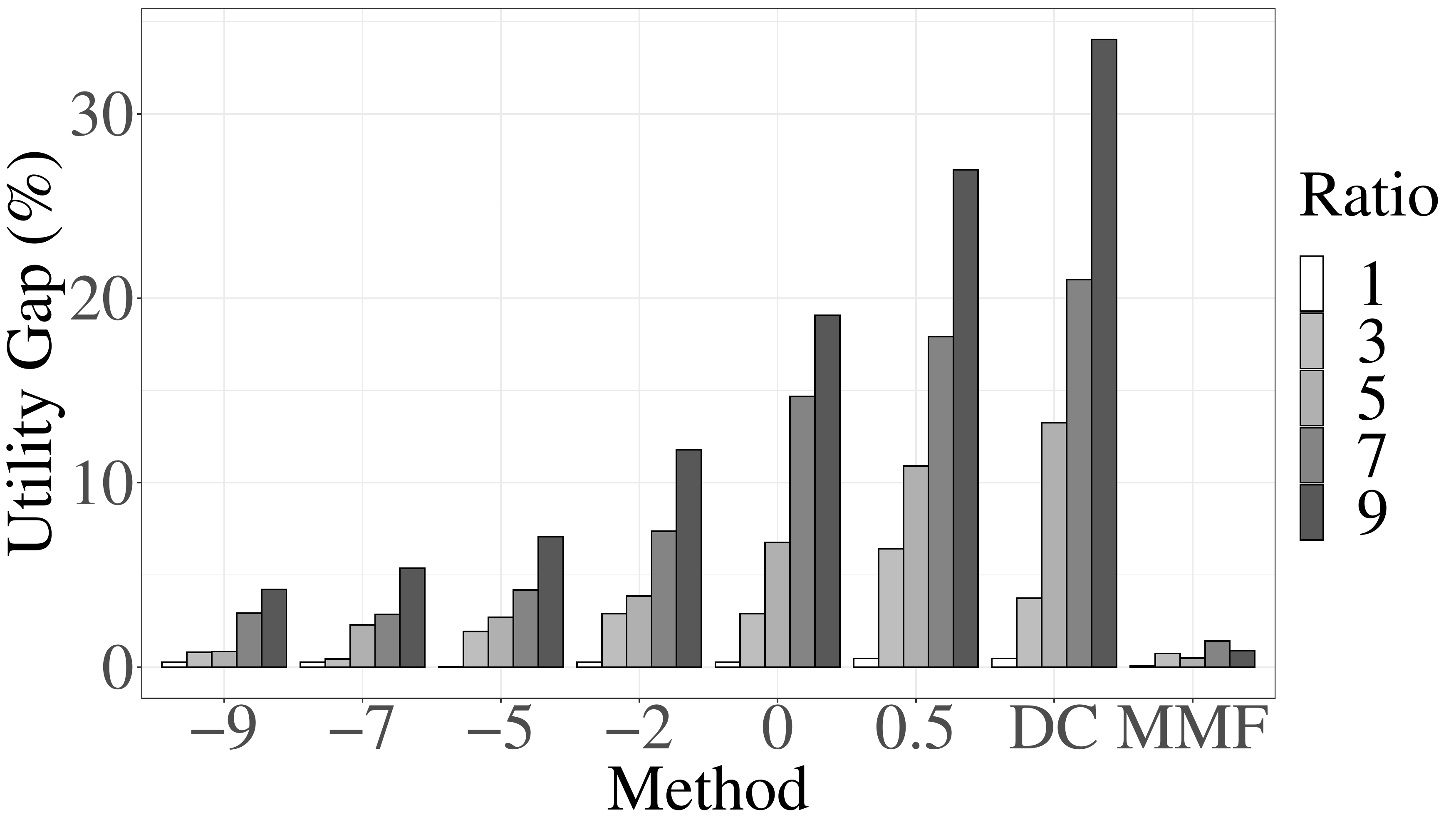}
\end{subfigure}%
\begin{subfigure}
\centering
\includegraphics[width = 0.32\textwidth]{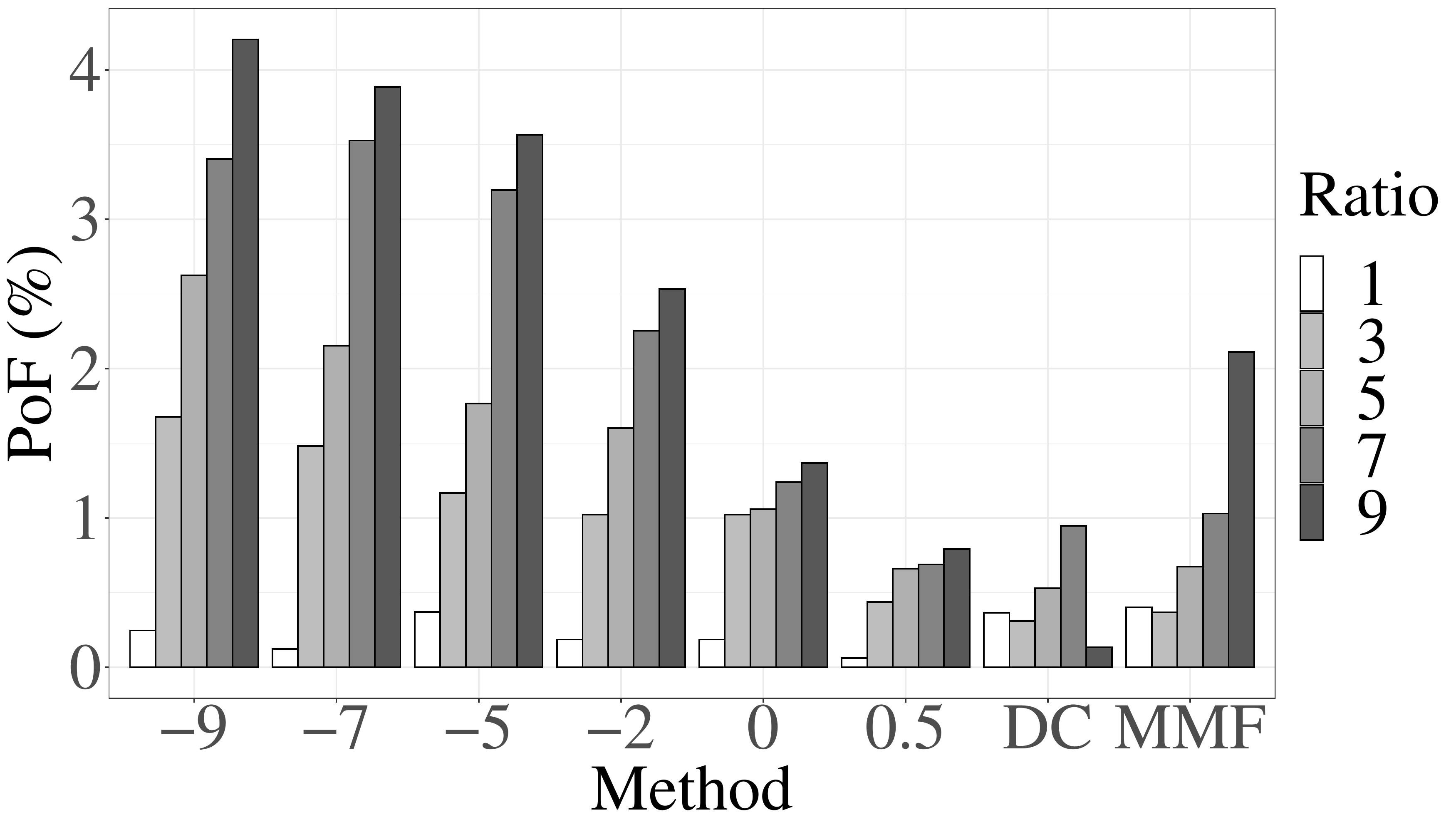}
\end{subfigure}
\caption{Utility gap and PoF for various relative community sizes where the ratio changes from 1 to 9. }
\label{fig:synthetic-all-size}
\end{figure}

In this section we study the effect of relative community size on both utility gap and efficiency.
We consider synthetic samples of SBM network consisting of two communities each of size 100 and we gradually increase the size of one community from 100 to 900 in order to study its effect on the utility values of each community. We set $q_\c = 0.005$ and $q_{\c\c'} = 0.001$. Results are summarized in Figure~\ref{fig:synthetic-all-size}. 
This result indicates that the utility gap increases with the relative community size, suggesting that minorities can be adversely affected without appropriate fairness measures in place. We also note that the strength of our approach is in its flexibility to trade-off fairness with efficiency. We may encounter scenarios where the fairness-efficiency trade-offs are mild (as in the particular setting of Figure~\ref{fig:synthetic-all-size}), but this does not undermine our approach as there are many practical situations (as discussed in real-world applications in the paper) where clearly this is not the case and our approach can handle all those cases effectively. DC exhibits relatively high utility gap. This is because by definition DC allocates more resources to communities that ``will do better with the resources'' and it does not always show an aversion to inequality, a result which we show theoretically in Section~\ref{sec:fair-connection}.

\subsection{Suicide Prevention Application}
\label{sec:exp-suicide}
\begin{table*}[ht!] 
  \centering
  \small
  \begin{tabular}{cccccccc}
    \toprule
    Network Name & $\#$ of Vertices & $\#$ of Edges& White & Black & Hispanic & Mixed Race & Other 
    \\
    \hline
    W1MFP & 219 & 217 & 16.4 & 41.5 & 20.5 & 16.4 & 5.0 
    \\
    W2MFP & 243 & 214 & 16.8 & 36.6  & 21.8 & 22.2 & 2.4  
    \\
    W3MFP & 296 & 326 & 22.6 & 34.4 & 15.2 & 22.9 & 4.7
    \\    
    W2SPY & 133 & 225 & 55.6 & 10.5  & -- & 22.5 & 11.3
    \\
    W3SPY & 144 & 227 & 63.0 & -- & -- & 16.0  &  20.0 
    \\
    W4SPY & 124 & 111 & 54.0 & 16.1 & -- & 14.5 & 15.3 
    \\
    \toprule
    \end{tabular}
  \caption{Racial composition (\%) after pre-processing as well as the number of vertices and edges of the social networks~\cite{barman2016sociometric}.}
  \label{table:racial_composition}
\end{table*}

\begin{table*}[ht!] 
  \centering
  \small
  \begin{tabular}{cccccccccc}
    \toprule
    \multirow{3}{*}{Measure (\%)} & \multicolumn{9}{c}{Fairness Approach}
    \\
    \cmidrule(lr){3-9}\morecmidrules\cmidrule(lr){3-9}
    & $\k$ & $\alpha = -5$ & $\alpha = -2$ & $\alpha = 0$ & $\alpha = 0.5$ & $\alpha = 0.9$ & DC & Maximin & IM
    \\
    \hline 
    \hline
    \multirow{6}{*}{utility gap} & 5 & 4.8 & 7.5 & 8.5 & 9.7 & 11.4 & 8.2 & \textbf{3.5} & 12.5
    \\
    & 10 & 4.6 & 6.6 & 7.3 & 9.5 & 12.9 & 6.9 & \textbf{2.0} & 11.7
    \\
    & 15 & 3.6 & 5.2  & 5.9 & 8.9 & 13.5 & 7.6 & \textbf{2.4} & 15.3
    \\
    & 20 & 3.6 & 4.4 & 5.9 & 7.3 & 14.0 & 5.8 & \textbf{2.3} & 17.2
    \\
    & 25 & 2.6 & 3.5  & 4.6 & 5.8 & 13.2 & 7.0 & \textbf{2.0} & 16.6
    \\
    & 30 & 2.4 & 3.2 & 4.3 & 6.4 & 8.6  & 8.2 & \textbf{2.0} & 15.7
    \\
    \hline
    \hline 
    \multirow{5}{*}{PoF} & 5 & 6.9 & 5.1 & 3.4 & 1.5 & \textbf{0.7} & 14.4 & 16.6 & 0.0
    \\
    & 10 & 6.6 & 3.8 & 2.8 & 0.7 & \textbf{0.1} & 14.3 & 13.1 & 0.0
    \\
    & 15 & 3.8 & 2.5 & 1.6 & 1.1 & \textbf{0.1} & 14.6 & 10.5 & 0.0
    \\
    & 20 & 4.6 & 3.8 & 2.9  & 2.0 & \textbf{1.0} & 13.9 & 10.5 & 0.0
    \\
    & 25 & 4.0 & 3.2 & 2.5 & 1.9 & \textbf{1.0} & 13.4 & 9.9 & 0.0
    \\
    & 30 & 3.9 & 3.4 & 2.9 & 2.3 & \textbf{1.8} & 12.7 & 10.8 & 0.0
    \\
    \bottomrule
    \end{tabular}
  \caption{Summary of the utility gap and PoF results averaged over 6 different real world social networks for various budget, fairness approaches and baselines. Numbers in bold highlight the best values in each setting (row) across different approaches.}
  \label{table:comparison_main}
  \vspace*{-\baselineskip} 
\end{table*}

Influence maximization has been previously implemented for health promoting interventions among homeless youth~\cite{YadavCXXRT16,wilder2020clinical}. In this section, we consider the problem of training a set of individuals who can share information to help prevent suicide (e.g., how to identify warning signs of suicide among others). We present simulation results over six different social networks of homeless youth from a major city in US as described in~\citet{barman2016sociometric}. We provide aggregate summaries of these networks (e.g., size, edge density and community statistics) in Table~\ref{table:racial_composition}. The data set consists of six different social networks of homeless youth from a major city in US as described in detail in~\citet{barman2016sociometric}. 
Each social network consists of 7 racial groups, namely, \emph{Black or African American, Latino or Hispanic, White,  American Indian or Alaska Native, Asian, Native Hawaiian or Other Pacific Islander and Mixed race}.  Each individual belongs to a single racial group. We use these partitioning by race to define our communities. However, to avoid misinterpretation of the results, we combine racial groups with a population $<10\%$ of the network size~$\n$ under the ``Other'' category. After this pre-processing step, each dataset will contain 3 to 5 communities. Results are summarized in Table~\ref{table:racial_composition}. We remark that the absent of a racial category in a given network is due to their small sizes and hence being merged into the ``Other" category after pre-processing (e.g., Hispanic in network W2SPY.)

We compare our welfare-based framework for different values of $\alpha$ against DC, MMF and influence maximization without fairness considerations (IM).
Table~\ref{table:comparison_main} provides a summary of the results averaged over all network instances where the numbers in bold highlight the best values (minimum utility gap and PoF) for each budget and across different fairness approaches. As seen, IM typically has a large utility gap (up to 17.2\% for $\k = 20$ which is significant because the total influence is only 28.40\%). By imposing fairness we can reduce this gap. In fact, we observe that across different values of $\alpha$ ranging from -5 to 0.5, there is a decreasing trend in utility gap, where
for $\k = 20$ and with $\alpha = -5$, we are able to decrease the utility gap by $3.6\%$.  Consistent with previous results on SBM networks, both MMF and $\alpha = -5$ exhibit very low utility gaps, however, MMF results in higher PoF. Furthermore, across the range of $\alpha$ we observe a mild trade-off between fairness and utility. This shows that in these networks enforcing fairness comes at a low cost, though as we see in the landslide setting, this is not always the case. 

\begin{figure}[ht!]
\centering
\begin{subfigure}
\centering
\includegraphics[width = 0.49\textwidth]{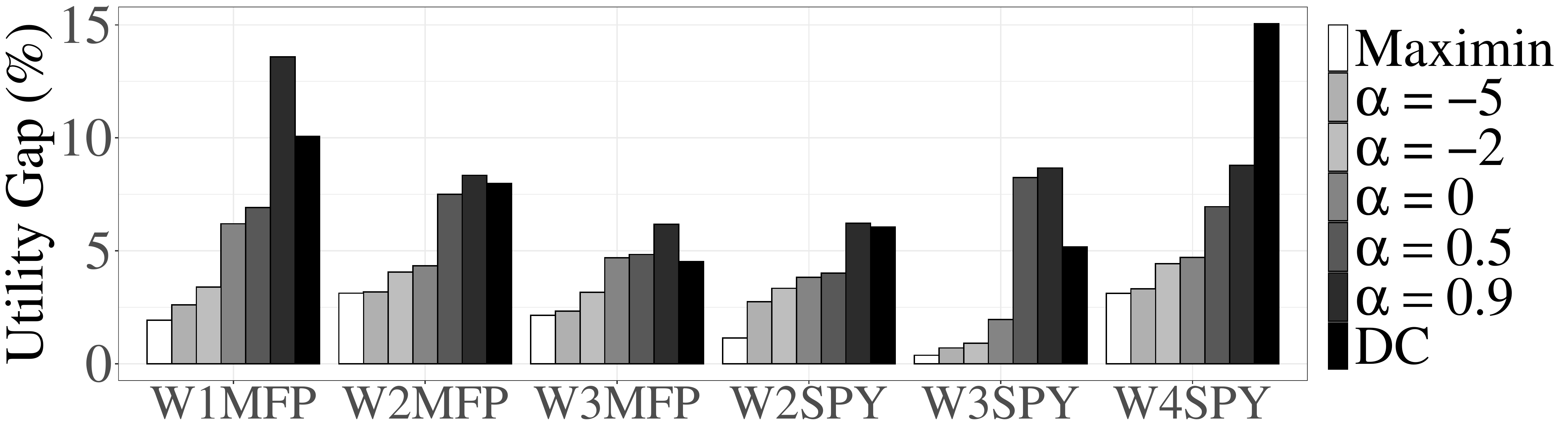}
\end{subfigure}
\begin{subfigure}
\centering
\includegraphics[width = 0.49\textwidth]{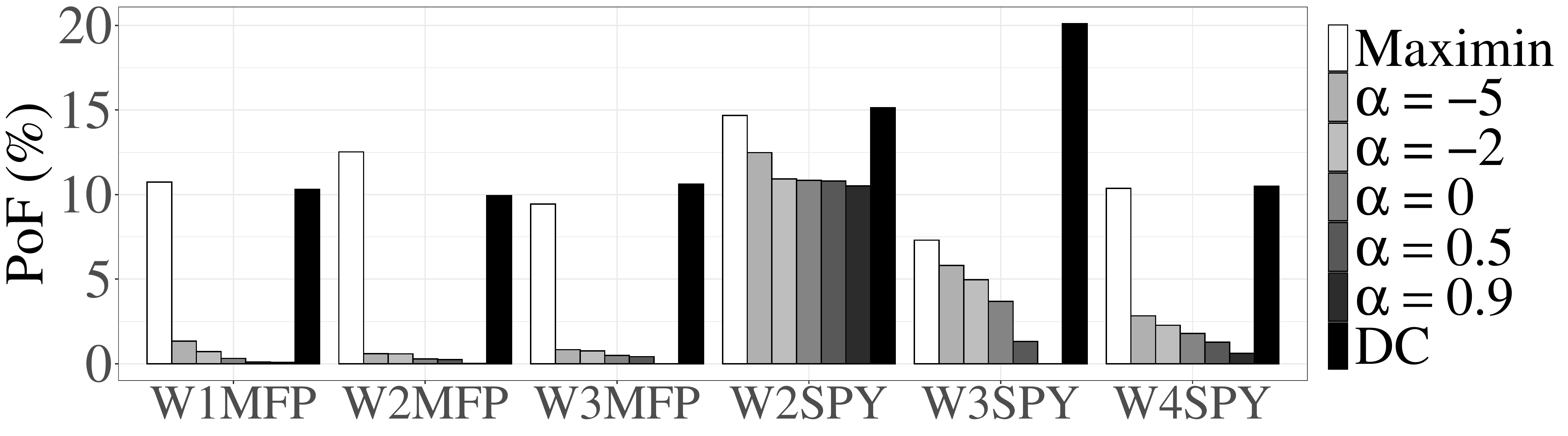}
\end{subfigure}
\caption{Top and bottom panels: utility gap and PoF for each real world network instances ($\k =30$).}
\label{fig:real-all}
\end{figure}

Figure~\ref{fig:real-all} shows the results for each network separately (X axis) for a fixed budget $\k=30$. Figure~\ref{fig:real-all} shows that the trade-offs can also be very network-dependent (compare e.g. W2SPY and W3MFP). This highlights the crucial need for a flexible framework that can be easily adjusted to meaningfully compare these trade-offs. 

\end{document}